\documentclass[showpacs,amsmath,amssymb,twocolumn,superscriptaddress,notitlepage,preprintnumbers,prl]{revtex4-1}

\usepackage{graphicx}
\usepackage{dcolumn}
\usepackage{bm}
\usepackage{times}
\usepackage{multirow}
\usepackage{dsfont}
\usepackage{dcolumn}
\usepackage{tabularx}
\usepackage{subfigure}
\usepackage{booktabs}
\usepackage{algorithmic}
\usepackage{ae}
\usepackage{lettrine}
\usepackage{color}
\usepackage{amsmath,amsthm,amssymb,amsfonts,boxedminipage}
\usepackage{epstopdf}
\usepackage{hyperref}
\usepackage{dsfont}
\usepackage{amsmath}
\newcommand{\algmargin}{\the\ALG@thistlm}

\makeatother

\newcommand{\norm}[1]{\left\lVert#1\right\rVert}

\newtheorem{theorem}{Theorem}
\newtheorem{lemma}{Lemma}
\newtheorem{corollary}{Corollary}
\newtheorem{definition}{Definition}

\newtheorem{remark}{Remark}

\makeatletter
\renewcommand\@biblabel[1]{#1.}
\makeatother

\usepackage[noend,noline,ruled,linesnumbered]{algorithm2e}
\usepackage{hyperref}
\usepackage{bm}
\usepackage{times}
\usepackage{multirow}
\usepackage{dcolumn}
\usepackage{tabularx}
\usepackage{graphicx}
\usepackage{subfigure}
\usepackage{ae}
\usepackage{lettrine}
\usepackage{color}
\usepackage{makecell}

\usepackage{amsmath,amsthm,amssymb,amsfonts,boxedminipage}
\usepackage{epstopdf}

\definecolor{comment}{rgb}{0, 0, 0}

\hypersetup{colorlinks=true, linkcolor=cyan, citecolor=cyan, urlcolor=blue }

\usepackage{braket}
\newcommand{\tr}[1]{\operatorname{Tr}\left(#1\right)}

\newcommand{\Cbb}{\mathbb{C}}

\newcommand{\Ord}[1]{\mathcal{O}\left( #1 \right)}
\newcommand{\abs}[1]{\left| #1 \right|}

\newcommand{\PKU}{Center on Frontiers of Computing Studies, School of Computer Science, Peking University, Beijing 100871, China}
\newcommand{\bnu}{School of Artificial Intelligence,
 Beijing Normal University, Beijing,
 100875, China}

\begin{document}

\title{Measuring Less to Learn More: Quadratic Speedup in learning Nonlinear Properties of Quantum Density Matrices}

\author{Yukun Zhang}
\email{yukunzhang@stu.pku.edu.cn}
\affiliation{\PKU}

\author{Yusen Wu}
\email{yusen.wu@bnu.edu.cn}
\affiliation{\bnu}

\author{You Zhou}
\email{you\_zhou@fudan.edu.cn}
\affiliation{Key Laboratory for Information Science of Electromagnetic Waves (Ministry of Education), Fudan University, Shanghai 200433, China}

\author{Xiao Yuan}
\email{xiaoyuan@pku.edu.cn}
\affiliation{\PKU}

\date{\today}

\begin{abstract}
A fundamental task in quantum information science is to measure nonlinear functionals of quantum states, such as $\tr{\rho^k O}$. Intuitively, one expects that computing a $k$-th order quantity generally requires $\Ord{k}$ copies of the state $\rho$, and we rigorously establish this lower bound under sample access to $\rho$. Surprisingly, this limitation can be overcome when one has purified access via a unitary that prepares a purification of $\rho$, a scenario naturally arising in quantum simulation and computation. In this setting, we find a different lower bound of $\Theta(\sqrt{k})$, and present a quantum algorithm that achieves this bound, demonstrating a quadratic advantage over sample-based methods. The key technical innovation lies in a designed quantum algorithm and optimal polynomial approximation theory---specifically, Chebyshev polynomial approximations tailored to the boundary behavior of power functions. 
Our results unveil a fundamental distinction between sample and purified access to quantum states, with broad implications for estimating quantum entropies and quantum Fisher information, realizing quantum virtual distillation and cooling, and evaluating other multiple nonlinear quantum observables with classical shadows.

\end{abstract}

\maketitle

Estimating nonlinear properties $\tr{\rho^k O}$ of quantum state $\rho$ and observable $O$ plays a pivotal role in quantum information science, with broad applications including quantum entropy estimation~\cite{ohya2004quantum}, quantum error mitigation via virtual distillation~\cite{Koczor2021exponential,Huggins2021virtual}, and probing thermal properties through quantum virtual cooling~\cite{cotler2019quantum}.
The standard approach to estimating such quantities---based on the generalized swap test or related protocols---requires $\Ord{k}$ copies of $\rho$~\cite{ekert2002direct,brun2004measuring}.
While recent techniques like shadow tomography~\cite{aaronson2018shadow} and classical shadows~\cite{huang2020predicting,zhou2024hybrid,grier2024principal,liu2024auxiliary} enable the efficient estimation of many observables simultaneously, they still require queries that scale linearly with $k$ and any attempts to reduce this scaling require exponential resources~\cite{chen2021hierarchy,chen2022exponential,chen2024optimal}. 
This persistent linear barrier raises a fundamental question: \emph{What are the ultimate limits and optimal strategies for estimating nonlinear properties of quantum states?}

In this work, we provide a systematic answer to this question by analyzing two standard access models for quantum states. In the sample access model, where identical copies of the state $\rho$ are available, we establish a lower bound of $\Omega(k)$ copies to estimate $\mathrm{Tr}(\rho^k O)$, confirming that conventional linear-scaling methods are optimal in this setting. Our proof relies on a reduction of the estimation task to a quantum state discrimination problem~\cite{helstrom1969quantum}. In sharp contrast, under purified quantum query access~\cite{gilyen2019quantum}, where a unitary prepares a purification of $\rho$ is given, we prove a different lower bound of $\Omega(\sqrt{k})$, extending fundamental approximation theories~\cite{sachdeva2014faster,trefethen2019approximation,montanaro2024quantum}. 
These distinct bounds indicate a fundamental separation between sample access and purified access to quantum states.

Next, we construct a tailored quantum algorithm that saturates the new bound. The algorithm achieves an additive error $\epsilon$ with query complexity $\Ord{\sqrt{k\log(\|O\|/\epsilon)}\|O\|/\epsilon}$. This separation arises from an unexpected connection between quantum algorithm design and classical approximation theory~\cite{sachdeva2014faster,trefethen2019approximation}. 
Moreover, we show that the decision version of nonlinear property estimation is BQP-complete, indicating the possibility of universal quantum speedups. These results have broad implications, including applications to entropy and quantum Fisher information estimation, virtual distillation for error mitigation, quantum virtual cooling, and accelerated classical shadow methods. 
Overall, our result reveals a fundamental---and perhaps surprising---separation between purified query access and sample access to quantum states, underscoring how different access to $\rho$ can dramatically alter the complexity of quantum estimation tasks.

\vspace{0.2cm}
\noindent\textbf{Main results---}We begin by formally defining the nonlinear property estimation problem:
\begin{definition}[Nonlinear Property Estimation]
\label{def:problem}
Given an $n$-qubit density matrix $\rho \in \Cbb^{2^n \times 2^n}$, an observable $O$, a positive integer $k \geq 2$, and parameters $\epsilon \in (0,1]$, the goal is to output an estimation $\widetilde{E}$ such that 
$
\abs{\widetilde{E} - \tr{\rho^k O}} \leq \epsilon$
with constant success probability.
\end{definition}

For sample access to $\rho$, a conventional method is to utilize the generalized swap test~\cite{ekert2002direct,brun2004measuring} inspired by the identity ${\rm Tr}\left(P_k\rho^{\otimes k}(O\otimes I_{k-1})\right)={\rm Tr}(\rho^k O)$, where $P_k$ represents a cyclic permutation operator enabling $P_k|a_1,a_2\cdots,a_k\rangle=|a_k,a_1,a_2,\cdots,a_{k-1}\rangle$. It is easy to check that it requires $\mathcal{O}(k\norm{O}^2/\epsilon^2)$ copies of $\rho$ with circuit depth $\Theta(k)$ and operator norm $\|O\| = \max_{\ket{v}}\braket{v|O|v}/\braket{v|v}$. For general sample access to $\rho$, we prove that the linear scaling in $k$ is necessary:


\begin{theorem}[Sample access lower bound]
\label{thm:lower}
There exists a density matrix $\rho$ and an observable $O$ such that any algorithm estimating $\tr{\rho^k O}$ to additive error $\epsilon$ with constant probability requires
$\Omega\left({k\norm{O}^2}/{\epsilon^2}\right)$
copies of $\rho$.
\end{theorem}
\noindent The lower bound verifies the optimality of the generalized swap test method. Here, we briefly review the idea of the proof and refer the reader to the End Matter for the details. We separately prove a lower bound on $k$, and then on $\epsilon$ and $\norm{O}$.

For proving the lower bound on the sample complexity $k$, we leverage fundamental links between sample complexity and quantum state discrimination~\cite{helstrom1969quantum}. That is, our construction encodes the quantum state discrimination problem into the nonlinear property estimation problem. The problem is set by properly chosen $\rho_0$, $\rho_1$ and $O$ such that estimating $\tr{\rho_i^k O},~i\in\{0,1\}$ to a prescribed accuracy discriminate the two cases. We assume the existence of an algorithm estimating the nonlinear property to the prescribed accuracy with a constant probability (say, at least ${2}/{3}$) using $m$ copies of the input state. Then, the Helstrom bound~\cite{helstrom1969quantum} dictates the best successful probability in distinguishing the two states, which relates to the trace distance between the two states. This then gives us the lower bound on the copies of $\Omega(k)$ to guarantee a constant success probability.

For the lower bounds on $\epsilon$ and $\norm{O}$, we lift the probability distributions distinguish problem into the quantum setting, the nonlinear property estimation problem. Again, by selecting $\rho_0$, $\rho_1$, and $O$ properly, estimating the nonlinear property then encodes the problem of discriminating between two probability distributions. Then, Le Cam’s two-point bound~\cite{lecam1973convergence,tsybakov2009introduction} suggests the hardness of distinguishing between the two distributions, resulting in the lower bound on the operator norm and accuracy.\\

On the other hand, the quantum purified query access to $\rho=\sum_{i=1}^N p_i\ket{\psi_i}\bra{\psi_i}$ assumes the existence of a unitary $U_\rho$ satisfying 
\begin{equation}\label{eq:purified access}
    U_\rho|0\rangle_E|0\rangle_I=\left|\rho\right\rangle_{E I}=\sum_{i=1}^N \sqrt{p_i}\left|\phi_i\right\rangle_E\left|\psi_i\right\rangle_I,
\end{equation}
where $\left\langle\phi_i | \phi_j\right\rangle=\left\langle\psi_i | \psi_j\right\rangle=\delta_{i j}$. It is a natural model when we consider subsystems of a quantum computer, as in Refs.~\cite{doi:10.1126/sciadv.aaz3666,doi:10.1126/science.aau4963,PhysRevLett.125.200501}.
Surprisingly, we show that the lower bound instead scales as $\Omega(\sqrt{k})$:

\begin{theorem}[Quantum purified query access lower bound]
\label{thm:lower2}
There exists a density matrix $\rho$ and an observable $O$ such that any algorithm estimating $\tr{\rho^k O}$ to additive error $\epsilon$ with constant probability requires
$\Omega({\sqrt{k\log(\norm{O}/\epsilon)}\norm{O}}/{\epsilon})$
queries to $U_\rho$.
\end{theorem}
\noindent The proof relies on establishing a new degree lower bound for approximating the power function using fundamental approximation theory~\cite{trefethen2019approximation}, and then lifting this lower bound to the query complexity lower bound for approximating matrix functions, leveraging the recent breakthrough of Ref.~\cite{montanaro2024quantum}. 

First, we prove the lower bound for approximating power functions:
\begin{lemma}[Optimal Approximation of Power Functions]
\label{thm:chebyshev}
For power function $x^k$ on the interval $[0,1]$ and $\epsilon > 0$, there exists a polynomial $Q(x)=p_d(x)$ of degree
$
d = \Theta\left(\sqrt{k\log(1/\epsilon)}\right)
$,
satsifying $\sup_{x \in [0,1]} |x^k - p_d(x)| \leq \epsilon$.
\end{lemma}
\noindent The quadratic improvement from $k$ to $\sqrt{k}$ can be heuristically understood from the Bernstein's inequality:
\begin{lemma}[Bernstein's Inequality~\cite{schaeffer1941inequalities}]
For a polynomial $p(x)$ of degree $\ell$ with $|p(x)| \leq 1$ on $[-1,1]$, the derivative satisfies:
\begin{itemize}
\item Interior: $|p'(x)| \leq \ell/\sqrt{1-x^2}$ for $x \in (-1,1)$
\item Boundary: $|p'(\pm 1)| \leq \ell^2$.
\end{itemize}
\end{lemma}
\noindent This reveals a fundamental asymmetry: polynomials can oscillate \emph{quadratically} faster at boundaries than at interior points. 
This heuristic argument captures why degree $\sqrt{k}$ suffices—the function's steepest behavior occurs precisely where polynomials are most powerful.
{Exploiting Ref.~\cite{montanaro2024quantum}, we rigorously prove that the query complexity for approximating matrix function $\rho^k$ is also lower bounded by $\Omega(k)$. We refer to the Supplementary Materials for the complete proof on the dependence on $\norm{O}$ and $\epsilon$.}
Similar insights have guided the design of other quantum algorithms for implementing different matrix functions. For example, quantum fast-forwarding achieves quadratic speedups in the mixing time of Markov Chain Monte Carlo~\cite{apers2018quantum}. In ground-state preparation of (nearly) frustration-free Hamiltonians~\cite{thibodeau2023nearly}, the dependence on the spectral gap is quadratically improved, echoing the observation made here. Related ideas have also enabled quadratic accelerations in certain quantum linear solvers~\cite{orsucci2021solving} and in algorithms for ordinary differential equations~\cite{an2022theory}.
However, the quadratic speedup for power functions of density matrices is established here for the first time.\\



The results suggest a fundamental distinction between sample access and purified quantum query access. To establish this separation rigorously, we further prove the following achievable result:
\begin{theorem}[Main Result]
\label{thm:main}
There exists a quantum algorithm that solves the nonlinear property estimation problem (Definition \ref{def:problem}) using
$\Ord{{\sqrt{k\log(\norm{O}/\epsilon)}\norm{O}}/{\epsilon}}
$
queries to $U_\rho$ and $U_\rho^\dagger$ with success probability at least $\frac{8}{\pi^2}$.
\end{theorem}

Before presenting the algorithm, we further prove that the decision version of the non-linear property estimation problem is BQP-complete.
\begin{theorem}[Informal]\label{thm:bqp}
Let $\rho$ be an $n$-qubit mixed state of rank $\mathcal{O}(1)$ with a succinct classical description of its purification. Let $O$ be an observable of operator norm $\norm{O}=\mathcal{O}(1)$ with a succinct classical description. Fix an integer $k\geq 1$. Then, promise either $\tr{\rho^k}\geq a$ or $\tr{\rho^k}\leq b$. For $a-b\geq 1/\mathrm{poly}(n)$, the promise problem is BQP-complete.
\end{theorem}
\noindent Here, the succinct classical description means that we have a description of the quantum circuit of size at most $\mathcal{O}(\mathrm{poly}(n))$.
The result further indicates the possibility of universal quantum speedups for general tasks.

\vspace{0.2cm}

In the following, we mainly focus on the proof of the main theorem.

\noindent\textbf{Quantum Algorithm---}Given the quantum purified query access~\cite{gilyen2019quantum}, we can block encode $\rho$ into a larger unitary 
\begin{equation}\nonumber
O_\rho =
\left(
\begin{array}{cc}
\rho & * \\ 
* & *
\end{array}
\right),
\end{equation}
which is realizable with one query to $U_\rho$ and $U_\rho^\dagger$. Suppose $O$ is Hermitian with real $\tr{\rho^kO}$. Given the optimal approximation of Lemma~\ref{thm:chebyshev}, our algorithm works as follow:
\begin{enumerate}
    \item Construct the controlled block encoding of $\hat{p}(\rho):=p(\rho)O$ with normalization factor $\alpha_O$ using the quantum singular value transformation (QSVT) algorithm~\cite{gilyen2019quantum}, where $p(x)$ is an optimal approximation of $g(x)=x^{k-1}$.
    \item Implement the quantum circuit as illustrated in Fig.~\ref{fig:1} with $W=I$ for estimating the real parts of the nonlinear property function.
    \item Apply amplitude estimation~\cite{brassard2000quantum} to estimate the probability $P(0)$ of the first ancilla qubit in state $\ket{0}$ to error $\frac{\epsilon}{2}$.
    \item The final estimation is given by
    \begin{equation}\label{eq:estimator}
        \widetilde{E}:=\alpha_O(2P({0})-1)).
    \end{equation}
\end{enumerate}
In step 1, $\hat{p}(\rho)$ is realized by first constructing a block encoding of $p(\rho)$ with error ${\epsilon}/{(2\norm{O})}$ using degree $m=\sqrt{2(k-1)\ln(4\alpha_O\norm{O}/\epsilon)}$ (see proof below). Then, we can block encode the product of $p(\rho)$ and $O$, which is a feature enabled by the QSVT framework. Here, we assume an error-free block encoding of $O$, which is true whenever $O$ is a linear combination of unitaries.
We denote the normalization factor of the block encoding of $O$ as $\alpha_O$ such that $\alpha_O\geq \norm{O}=\mathcal{O}(\norm{O})$.

\begin{figure}[t]
\centering
\includegraphics[width=0.5\textwidth]{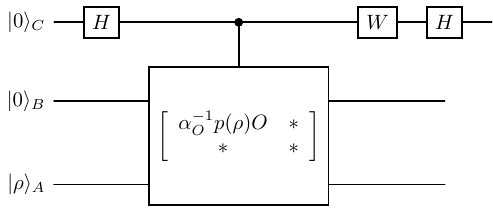} 
\caption{The quantum circuit for estimating $\tr{\rho^k O}$. The controlled operation is a unitary that block encodes the $p(\rho)O$, where $p(x)$ is a polynomial function that approximates $x^{k-1}$ and $\alpha_O$ is a normalization factor of $O$.
When setting $W=I$ or $S^\dagger$ with $S$ the phase gate, the final probability of measuring $\ket{0}_C$ for the first ancilla register encodes information about the real or imaginary parts of the targeted quantity. Specifically, the estimator is given by Eq.~\eqref{eq:estimator} for either the real or the imaginary part. Therefore, one performs amplitude estimation on the register $C$ to estimate the probability.
}
\label{fig:1}
\end{figure}

We can easily verify that the probability of measuring the state $\ket{0}$ is given by $[{1+\mathrm{Re}(\tr{\rho p(\rho)O})/\alpha_O}]/{2}$,
where $\mathrm{Re}(\cdot)$ denotes taking the real part. Since $p(\rho)\approx \rho^{k-1}$ and $\tr{\rho^kO})$ is real, it proves that Eq.~\eqref{eq:estimator} is a good approximation to $\tr{\rho^k O}$.
Furthermore, for non-hermitian $O$, we can set $W=S^\dagger$ with the phase gate $S$ to estimate the imaginary part.

Now, we are ready to prove our main result, Theorem~\ref{thm:main}.
\begin{proof}
The total error decomposes as the sum of the approximation error in polynomial approximation and the estimation error in amplitude estimation $
        \abs{\widetilde{E} - \tr{\rho^k O}} \leq \abs{\widetilde{E} - \tr{\rho p'(\rho)}}+ \abs{\tr{\rho p'(\rho)}- \tr{\rho^k O}}$
where each source of error is set to be $\epsilon/2$. 

For the approximation error $\abs{\widetilde{E} - \tr{\rho p'(\rho)}}$, we can apply Lemma~\ref{thm:chebyshev} and prove
\begin{equation}\nonumber
\begin{aligned}
\abs{\tr{\rho p(\rho)O} - \tr{\rho^k O}} &\leq \norm{\rho}_1\cdot\norm{p(\rho)O-\rho^{k-1}O} \\
&\leq \norm{O} \cdot \norm{p(\rho) - \rho^{k-1}} \\
&\leq \norm{O} \cdot \frac{\epsilon}{2\norm{O}} = \frac{\epsilon}{2}.
\end{aligned}
\end{equation}
In the first line, we apply H\"{o}lder's inequality and $\|\rho\|_1$ denotes the Schatten one norm. In the second line, we have used the sub-multiplicativity of the operator norm. This error originates from approximating $g(x)=x^{k-1}$ with $p(x)$, which is $\epsilon/(2\norm{O})$. Based on Theorem~\ref{thm:chebyshev}, a degree $\sqrt{(k-1)\log(1/\varepsilon)}$ polynomial function gives an approximation error $\varepsilon$ of $g(x)=x^{k-1}$. Therefore, the desirable degree is given by $d=\Ord{\sqrt{k\log(\norm{O}/\epsilon)}}$ with $\alpha_O=\mathcal{O}(\norm{O})$.


The estimation accuracy is given by $\epsilon'=\frac{\epsilon}{4\alpha_O}$ as the estimation is amplified by a factor $2\alpha_O$ with $\alpha_O=\mathcal{O}(\norm{O})$. By standard results from amplitude estimation~\cite{brassard2000quantum}, $\mathcal{O}(1/\epsilon')$ queries to the $p'(\rho)$ circuit are needed with success probability $\frac{8}{\pi^2}$. Together, it proves that a total of $m=\Ord{\sqrt{k\log(\norm{O}/\epsilon)}\norm{O}/\epsilon}$ queries to $U_\rho$ is sufficient for estimating $\tr{\rho^kO}$.
\end{proof}
\noindent We note that implementing a controlled block encoding of $p'(\rho)\approx\rho^{k-1}O$ is essential. In contrast, approximately preparing the subnormalized state $\rho^k$~\cite{wang2024new} and estimating the nonlinear function through the state would incur a query complexity that scales with the rank of $\rho$, which can be exponentially large in general. This technique may also prove useful in overcoming rank dependence for other related tasks.

\vspace{0.2cm}

\noindent\textbf{Applications---}Our results have immediate applications in different fields.

In the special case with $O$ being identity, the nonlinear measurements can be used to estimate the Schatten norm $\norm{\rho}_p:=\left(\tr{\rho^p}\right)^{1/p}$ with $p\in\mathbb{Z}$, a basic quantifier in quantum information~\cite{brandao2015estimating}. Relatedly, one can compute entropic measures, such as the Rényi entropy: $S_\alpha(\rho)=\frac{1}{1-\alpha}\log(\tr{\rho^\alpha})$ with $\alpha>0$ and $\alpha\neq 1$, or the Tsallis entropy: $S_q(\rho)=\frac{\tr{\rho^q}}{1-q}$ with $\alpha\in\mathbb{Z}$, $q\in\mathbb{Z}$, and $q>2$.
Moreover, by employing the replica trick, one may approximate the min-entropy $S_{\infty}(\rho)=-\ln \|\rho\|_{\infty}$ or the von Neumann entropy $S(\rho) = \tr{\rho\log\rho}$ by taking the limit of $k\rightarrow\infty$ and $k\rightarrow1$, respectively, of which have been widely used in quantum field theory~\cite{calabrese2004entanglement,calabrese2005evolution} and gravity~\cite{faulkner2013quantum}.

Our method also applies to more sophisticated quantum information quantities, such as the quantum Fisher information (QFI), a key concept in quantum metrology~\cite{braunstein1994statistical,toth2014quantum} that sets the ultimate precision limit for parameter estimation via the quantum Cramér–Rao bound. The QFI also plays a broad role in quantum information theory and many-body physics, including quantifying the quantum speed limit~\cite{taddei2013quantum} and multipartite entanglement~\cite{hyllus2012fisher}. Given the eigen-decompostion of $\rho=\sum_ip_i\ket{\psi_i}\bra{\psi_i}$, the QFI is defined as $F_Q=2 \sum_{p_k+p_l>0} \frac{\left(p_k-p_l\right)^2}{p_k+p_l}|\langle \psi_k| O| \psi_l\rangle|^ 2$
with $O$ a Hermitian observable. Direct estimation is challenging since the formula is expressed in the eigenbasis of $\rho$. Recent work proposes Krylov subspace methods~\cite{zhang2025krylov} that approximate the QFI by measuring $f_\rho=\mathrm{Tr}(O\rho^l O \rho^{k-l+2})$ for $l=0,\dots,k+2$. Although these differ from the standard nonlinear quantities $\mathrm{Tr}(\rho^k O)$ considered above, our algorithm can be adapted by block encoding $O\rho^l O \rho^{k-l+1}$ for small or large $l$, respectively. This extension still yields a quadratic speedup in estimating the QFI. Besides, Ref.~\cite{rath2021quantum} provides the power function of the density matrix related formula, such that as the degree of the power function increases, the output monotonically approaches the true QFI.




Nonlinear properties of density matrices also play a critical role in quantum error mitigation, specifically virtual distillation (VD)~\cite{Koczor2021exponential, Huggins2021virtual}, which stands out as a particularly effective approach for mitigating errors in quantum state preparation and quantum chemistry computation~\cite{o2023purification}. The key insight of VD is that by effectively computing $\rho^k/\operatorname{tr}(\rho^k)$ for a mixed state $\rho$, one can amplify the contribution of the largest eigenvalue component, effectively ``purifying" the quantum state. This is especially useful when the mixed state $\rho$ is a noise-corrupted version of a pure target state. Therefore, for any observable $O$, it is approximated via VD as $
\langle O \rangle_{\text{distilled}} = {\tr{\rho^k O}}/{\tr{\rho^k}}$. The formulation has also been applied in extracting dominant eigenproperties of density matrices~\cite{seif2023shadow,hakoshima2024localized,bako2025exponential}. 
Meanwhile, quantum virtual cooling~\cite{cotler2019quantum} provides an efficient method for predicting thermal states at low temperature by using high-temperature thermal states. Specifically, for a thermal state $\rho_{\beta}=e^{-\beta H}/{\rm Tr}(e^{-\beta H})$ associated with $H$ and an observable $O$, the virtual cooling achieves the map
$    \rho^{\otimes m}\mapsto{{\rm Tr}(Oe^{-k\beta H})}/{{\rm Tr}(e^{-k\beta H})}$.
In both scenarios, our algorithm can output estimated values for the numerator and the denominator within $\epsilon$-additive error separately. The ratio estimation requires careful error propagation analysis
\begin{equation}
\left|\frac{a}{b} - \frac{\tilde{a}}{\tilde{b}}\right| \leq \frac{|a-\tilde{a}|}{b} + \frac{|a||b-\tilde{b}|}{b^2},
\end{equation}
where $\tilde{a}$ and $\tilde{b}$ are estimates of $a$ and $b$ respectively. Overall, our algorithm provides an efficient way for realizing VD and virtual cooling with $\Ord{\sqrt{k}}$ queries to $U_\rho$.

Finally, our method is compatible with the classical shadow framework~\cite{huang2020predicting} for estimating multiple observables. For $M$ observables, classical shadows provide an efficient strategy by postprocessing measurement outcomes
from random basis measurements of sequentially prepared states with sample complexity bounded by $\mathcal{O}\left(k\log(M)\max\{\|O_m\|_{\infty}^24^{\abs{{\rm supp}(O_m)}}\}\epsilon^{-2}\right)$.
To integrate our approach, we replace amplitude estimation with a Hadamard-test circuit, block encoding a polynomial approximation of $\rho^{k-1}$, and then performing randomized measurements on the system register.
The standard postprocessing of classical shadows then applies, reducing the complexity to  $\Ord{\sqrt{k}\log(M)\max\{\|O_m\|_{\infty}^24^{\abs{{\rm supp}(O_m)}}\}\epsilon^{-2}}$. Specifically, introducing a random Clifford $U$ on the system register followed by a computational-basis measurement yields~\cite{zhou2024hybrid}:
\begin{equation}
    \begin{aligned}
    &\operatorname{Tr}\big(X_C\otimes I_B\otimes I_E\otimes|b\rangle\langle b|_I U \left({\rm c}\text{-}U_p\right) \\&\quad(\ket{+}\bra{+}_C\otimes \ket{0^m}\bra{0^m}_B \otimes\ket{\rho}
    \bra{\rho}_{EI})\left({\rm c}\text{-}U_p\right)^\dagger\, U^\dagger \big)\\
= & \frac{1}{2}\left\{\operatorname{Tr}\left[X_C|1\rangle\langle 0|_C\right] \operatorname{Tr}\left[\langle b| \bra{0^m}UU_p \rho U^{\dagger}|b\rangle_I\ket{0^m}_B\right]\right. \\
& \left.+\operatorname{Tr}\left[X_C|0\rangle\langle 1|_C\right] \operatorname{Tr}\left[\langle b|\bra{0^m}UU_p  \rho  U^{\dagger}|b\rangle_I\ket{0^m}_B\right]\right\} \\
= & \langle b| U \rho p(\rho) U^{\dagger}|b\rangle\approx\langle b| U \rho^k U^{\dagger}|b\rangle,
    \end{aligned}
\end{equation}
where $C$ is a single-qubit register, $B$ is the ancillary register for the block encoding unitary $U_p$ (with ${\rm c}\text{-}U_p$ the controlled version, and the control register is $C$) that encodes $p(\rho)$, which approximates $\rho^{k-1}$, $E$ is the dilated register for purifying $\rho$, and $I$ is the system register.

This implies that $\mathcal{M}^{-1}(U^{\dagger}|b\rangle\langle b|U)$, where $\mathcal{M}=\mathbb{E}_{U,b}[(U^{\dagger}|b\rangle\langle b|U]$, serves as an unbiased estimator for (approximately) $\rho^k$. Consequently, our method achieves a quadratic improvement in $k$ over previous shadow-based approaches~\cite{zhou2024hybrid,liu2024auxiliary,li2024nearly}.

\vspace{0.2cm}
\noindent\textbf{Conclusion---}In this work, we address the fundamental problem of determining the optimal strategy for estimating nonlinear quantum properties. We establish two lower bounds: a linear bound in the standard sample-access model and a quadratic bound in the purified quantum query model.
Building on these results, we develop a novel quantum algorithm based on optimal polynomial approximations that saturates the quadratic lower bound, thereby achieving a genuine quadratic speedup over all previously known methods for nonlinear property estimation.
Finally, we demonstrate the broad applicability of our framework to important tasks, including the estimation of quantum entropies, quantum Fisher information, quantum virtual distillation, quantum virtual cooling, and classical shadow tomography.


Our work opens new avenues for quantum algorithms that fundamentally rely on nonlinear state properties. An intriguing direction for future research is to further explore separations between query access and sampling access to $\rho$. Prior works~\cite{wang2024optimal,wang2024sample} have revealed certain behavioral differences between the two settings, but most known advantages of query access over sampling access are limited to improvements in accuracy, largely enabled by amplitude estimation. A deeper investigation into the boundaries and potential separations between these two paradigms promises to be highly fruitful. In addition, we note that many algorithms for matrix functions of $\rho$ depend on the QSVT framework, with block encoding as a key primitive. Yet methods for constructing block encodings from sampling access are scarce, with density matrix exponentiation (DME)~\cite{lloyd2014quantum} being the most prominent. Developing broader and more efficient block encoding techniques beyond DME thus represents another important future challenge.


\vspace{0.2cm}

\begin{acknowledgements}
\noindent \textbf{Acknowledgments}---This work is supported by the National Natural Science Foundation of China Grant (No.~12361161602), NSAF (Grant No.~U2330201), the Innovation Program for Quantum Science and Technology (Grant No.~2023ZD0300200), and the High-performance Computing Platform of Peking University. YZ is supported by the Innovation Program for Quantum Science and Technology Grant Nos.~2024ZD0301900 and 2021ZD0302000, the National Natural Science Foundation of China (NSFC) Grant No.~12205048, the Shanghai Science and Technology Innovation Action Plan Grant No.~24LZ1400200, Shanghai Pilot Program for Basic Research - Fudan University 21TQ1400100 (25TQ003), and the start-up funding of Fudan University.
\end{acknowledgements}

\bibliography{ref}

\begin{thebibliography}{10}

\bibitem{ohya2004quantum}
Masanori Ohya and D{\'e}nes Petz.
\newblock {\em Quantum entropy and its use}.
\newblock Springer Science \& Business Media, 2004.

\bibitem{Koczor2021exponential}
B{\'a}lint Koczor.
\newblock Exponential error suppression for near-term quantum devices.
\newblock {\em Physical Review X}, 11(3):031057, 2021.

\bibitem{Huggins2021virtual}
William~J Huggins, Sam McArdle, Thomas~E O’Brien, Joonho Lee, Nicholas~C Rubin, Sergio Boixo, K~Birgitta Whaley, Ryan Babbush, and Jarrod~R McClean.
\newblock Virtual distillation for quantum error mitigation.
\newblock {\em Physical Review X}, 11(4):041036, 2021.

\bibitem{cotler2019quantum}
Jordan Cotler, Soonwon Choi, Alexander Lukin, Hrant Gharibyan, Tarun Grover, M~Eric Tai, Matthew Rispoli, Robert Schittko, Philipp~M Preiss, Adam~M Kaufman, et~al.
\newblock Quantum virtual cooling.
\newblock {\em Physical Review X}, 9(3):031013, 2019.

\bibitem{ekert2002direct}
Artur~K Ekert, Carolina~Moura Alves, Daniel~KL Oi, Micha{\l} Horodecki, Pawe{\l} Horodecki, and Leong~Chuan Kwek.
\newblock Direct estimations of linear and nonlinear functionals of a quantum state.
\newblock {\em Physical review letters}, 88(21):217901, 2002.

\bibitem{brun2004measuring}
Todd~A Brun.
\newblock Measuring polynomial functions of states.
\newblock {\em arXiv preprint quant-ph/0401067}, 2004.

\bibitem{aaronson2018shadow}
Scott Aaronson.
\newblock Shadow tomography of quantum states.
\newblock In {\em Proceedings of the 50th annual ACM SIGACT symposium on theory of computing}, pages 325--338, 2018.

\bibitem{huang2020predicting}
Hsin-Yuan Huang, Richard Kueng, and John Preskill.
\newblock Predicting many properties of a quantum system from very few measurements.
\newblock {\em Nature Physics}, 16(10):1050--1057, 2020.

\bibitem{zhou2024hybrid}
You Zhou and Zhenhuan Liu.
\newblock A hybrid framework for estimating nonlinear functions of quantum states.
\newblock {\em npj Quantum Information}, 10(1):62, 2024.

\bibitem{grier2024principal}
Daniel Grier, Hakop Pashayan, and Luke Schaeffer.
\newblock Principal eigenstate classical shadows.
\newblock In {\em The Thirty Seventh Annual Conference on Learning Theory}, pages 2122--2165. PMLR, 2024.

\bibitem{liu2024auxiliary}
Qing Liu, Zihao Li, Xiao Yuan, Huangjun Zhu, and You Zhou.
\newblock Auxiliary-free replica shadow estimation.
\newblock {\em arXiv preprint arXiv:2407.20865}, 2024.

\bibitem{chen2021hierarchy}
Sitan Chen, Jordan Cotler, Hsin-Yuan Huang, and Jerry Li.
\newblock A hierarchy for replica quantum advantage.
\newblock {\em arXiv preprint arXiv:2111.05874}, 2021.

\bibitem{chen2022exponential}
Sitan Chen, Jordan Cotler, Hsin-Yuan Huang, and Jerry Li.
\newblock Exponential separations between learning with and without quantum memory.
\newblock In {\em 2021 IEEE 62nd Annual Symposium on Foundations of Computer Science (FOCS)}, pages 574--585. IEEE, 2022.

\bibitem{chen2024optimal}
Sitan Chen, Jerry Li, and Allen Liu.
\newblock An optimal tradeoff between entanglement and copy complexity for state tomography.
\newblock In {\em Proceedings of the 56th Annual ACM Symposium on Theory of Computing}, pages 1331--1342, 2024.

\bibitem{helstrom1969quantum}
Carl~W Helstrom.
\newblock Quantum detection and estimation theory.
\newblock {\em Journal of Statistical Physics}, 1:231--252, 1969.

\bibitem{gilyen2019quantum}
Andr{\'a}s Gily{\'e}n, Yuan Su, Guang~Hao Low, and Nathan Wiebe.
\newblock Quantum singular value transformation and beyond: exponential improvements for quantum matrix arithmetics.
\newblock In {\em Proceedings of the 51st annual ACM SIGACT symposium on theory of computing}, pages 193--204, 2019.

\bibitem{sachdeva2014faster}
Sushant Sachdeva, Nisheeth~K Vishnoi, et~al.
\newblock Faster algorithms via approximation theory.
\newblock {\em Foundations and Trends{\textregistered} in Theoretical Computer Science}, 9(2):125--210, 2014.

\bibitem{trefethen2019approximation}
Lloyd~N Trefethen.
\newblock {\em Approximation theory and approximation practice, extended edition}.
\newblock SIAM, 2019.

\bibitem{montanaro2024quantum}
Ashley Montanaro and Changpeng Shao.
\newblock Quantum and classical query complexities of functions of matrices.
\newblock In {\em Proceedings of the 56th Annual ACM Symposium on Theory of Computing}, pages 573--584, 2024.

\bibitem{lecam1973convergence}
Lucien LeCam.
\newblock Convergence of estimates under dimensionality restrictions.
\newblock {\em The Annals of Statistics}, pages 38--53, 1973.

\bibitem{tsybakov2009introduction}
Alexandre~B. Tsybakov.
\newblock {\em Introduction to Nonparametric Estimation}.
\newblock Springer Series in Statistics. Springer, New York, NY, 2009.

\bibitem{doi:10.1126/sciadv.aaz3666}
Andreas Elben, Jinlong Yu, Guanyu Zhu, Mohammad Hafezi, Frank Pollmann, Peter Zoller, and Benoit Vermersch.
\newblock Many-body topological invariants from randomized measurements in synthetic quantum matter.
\newblock {\em Science Advances}, 6(15):eaaz3666, 2020.

\bibitem{doi:10.1126/science.aau4963}
Tiff Brydges, Andreas Elben, Petar Jurcevic, Benoit Vermersch, Christine Maier, Ben~P. Lanyon, Peter Zoller, Rainer Blatt, and Christian~F. Roos.
\newblock Probing renyi entanglement entropy via randomized measurements.
\newblock {\em Science}, 364(6437):260--263, 2019.

\bibitem{PhysRevLett.125.200501}
Andreas Elben, Richard Kueng, Hsin-Yuan~(Robert) Huang, Rick van Bijnen, Christian Kokail, Marcello Dalmonte, Pasquale Calabrese, Barbara Kraus, John Preskill, Peter Zoller, and Beno\^{\i}t Vermersch.
\newblock Mixed-state entanglement from local randomized measurements.
\newblock {\em Phys. Rev. Lett.}, 125:200501, Nov 2020.

\bibitem{schaeffer1941inequalities}
AC~Schaeffer.
\newblock Inequalities of a. markoff and s. bernstein for polynomials and related functions.
\newblock 1941.

\bibitem{apers2018quantum}
Simon Apers and Alain Sarlette.
\newblock Quantum fast-forwarding: Markov chains and graph property testing.
\newblock {\em arXiv preprint arXiv:1804.02321}, 2018.

\bibitem{thibodeau2023nearly}
Matthew Thibodeau and Bryan~K Clark.
\newblock Nearly-frustration-free ground state preparation.
\newblock {\em Quantum}, 7:1084, 2023.

\bibitem{orsucci2021solving}
Davide Orsucci and Vedran Dunjko.
\newblock On solving classes of positive-definite quantum linear systems with quadratically improved runtime in the condition number.
\newblock {\em Quantum}, 5:573, 2021.

\bibitem{an2022theory}
Dong An, Jin-Peng Liu, Daochen Wang, and Qi~Zhao.
\newblock A theory of quantum differential equation solvers: limitations and fast-forwarding.
\newblock {\em arXiv preprint arXiv:2211.05246}, 2022.

\bibitem{brassard2000quantum}
Gilles Brassard, Peter Hoyer, Michele Mosca, and Alain Tapp.
\newblock Quantum amplitude amplification and estimation.
\newblock {\em arXiv preprint quant-ph/0005055}, 2000.

\bibitem{wang2024new}
Qisheng Wang, Ji~Guan, Junyi Liu, Zhicheng Zhang, and Mingsheng Ying.
\newblock New quantum algorithms for computing quantum entropies and distances.
\newblock {\em IEEE Transactions on Information Theory}, 2024.

\bibitem{brandao2015estimating}
Fernando~GSL Brandao and Aram~W Harrow.
\newblock Estimating operator norms using covering nets.
\newblock {\em arXiv preprint arXiv:1509.05065}, 2015.

\bibitem{calabrese2004entanglement}
Pasquale Calabrese and John Cardy.
\newblock Entanglement entropy and quantum field theory.
\newblock {\em Journal of statistical mechanics: theory and experiment}, 2004(06):P06002, 2004.

\bibitem{calabrese2005evolution}
Pasquale Calabrese and John Cardy.
\newblock Evolution of entanglement entropy in one-dimensional systems.
\newblock {\em Journal of Statistical Mechanics: Theory and Experiment}, 2005(04):P04010, 2005.

\bibitem{faulkner2013quantum}
Thomas Faulkner, Aitor Lewkowycz, and Juan Maldacena.
\newblock Quantum corrections to holographic entanglement entropy.
\newblock {\em Journal of High Energy Physics}, 2013(11):1--18, 2013.

\bibitem{braunstein1994statistical}
Samuel~L Braunstein and Carlton~M Caves.
\newblock Statistical distance and the geometry of quantum states.
\newblock {\em Physical Review Letters}, 72(22):3439, 1994.

\bibitem{toth2014quantum}
G{\'e}za T{\'o}th and Iagoba Apellaniz.
\newblock Quantum metrology from a quantum information science perspective.
\newblock {\em Journal of Physics A: Mathematical and Theoretical}, 47(42):424006, 2014.

\bibitem{taddei2013quantum}
M{\'a}rcio~M Taddei, Bruno~M Escher, Luiz Davidovich, and Ruynet~L de~Matos~Filho.
\newblock Quantum speed limit for physical processes.
\newblock {\em Physical review letters}, 110(5):050402, 2013.

\bibitem{hyllus2012fisher}
Philipp Hyllus, Wies{\l}aw Laskowski, Roland Krischek, Christian Schwemmer, Witlef Wieczorek, Harald Weinfurter, Luca Pezz{\'e}, and Augusto Smerzi.
\newblock Fisher information and multiparticle entanglement.
\newblock {\em Physical Review A—Atomic, Molecular, and Optical Physics}, 85(2):022321, 2012.

\bibitem{zhang2025krylov}
Da-Jian Zhang and DM~Tong.
\newblock Krylov shadow tomography: Efficient estimation of quantum fisher information.
\newblock {\em Physical Review Letters}, 134(11):110802, 2025.

\bibitem{rath2021quantum}
Aniket Rath, Cyril Branciard, Anna Minguzzi, and Beno{\^\i}t Vermersch.
\newblock Quantum fisher information from randomized measurements.
\newblock {\em Physical Review Letters}, 127(26):260501, 2021.

\bibitem{o2023purification}
Thomas~E O’Brien, G~Anselmetti, Fotios Gkritsis, VE~Elfving, Stefano Polla, William~J Huggins, Oumarou Oumarou, Kostyantyn Kechedzhi, Dmitry Abanin, Rajeev Acharya, et~al.
\newblock Purification-based quantum error mitigation of pair-correlated electron simulations.
\newblock {\em Nature Physics}, 19(12):1787--1792, 2023.

\bibitem{seif2023shadow}
Alireza Seif, Ze-Pei Cian, Sisi Zhou, Senrui Chen, and Liang Jiang.
\newblock Shadow distillation: Quantum error mitigation with classical shadows for near-term quantum processors.
\newblock {\em PRX Quantum}, 4(1):010303, 2023.

\bibitem{hakoshima2024localized}
Hideaki Hakoshima, Suguru Endo, Kaoru Yamamoto, Yuichiro Matsuzaki, and Nobuyuki Yoshioka.
\newblock Localized virtual purification.
\newblock {\em Physical Review Letters}, 133(8):080601, 2024.

\bibitem{bako2025exponential}
Bence Bak{\'o}, Tenzan Araki, and B{\'a}lint Koczor.
\newblock Exponential distillation of dominant eigenproperties.
\newblock {\em arXiv preprint arXiv:2506.04380}, 2025.

\bibitem{li2024nearly}
Zihao Li, Changhao Yi, You Zhou, and Huangjun Zhu.
\newblock Nearly query-optimal classical shadow estimation of unitary channels.
\newblock {\em arXiv preprint arXiv:2410.14538}, 2024.

\bibitem{wang2024optimal}
Qisheng Wang.
\newblock Optimal trace distance and fidelity estimations for pure quantum states.
\newblock {\em IEEE Transactions on Information Theory}, 2024.

\bibitem{wang2024sample}
Qisheng Wang and Zhicheng Zhang.
\newblock Sample-optimal quantum estimators for pure-state trace distance and fidelity via samplizer.
\newblock {\em arXiv preprint arXiv:2410.21201}, 2024.

\bibitem{lloyd2014quantum}
Seth Lloyd, Masoud Mohseni, and Patrick Rebentrost.
\newblock Quantum principal component analysis.
\newblock {\em Nature physics}, 10(9):631--633, 2014.

\bibitem{fuchs2002cryptographic}
Christopher~A Fuchs and Jeroen Van De~Graaf.
\newblock Cryptographic distinguishability measures for quantum-mechanical states.
\newblock {\em IEEE Transactions on Information Theory}, 45(4):1216--1227, 2002.

\bibitem{kitaev2002classical}
Alexei~Yu Kitaev, Alexander Shen, and Mikhail~N Vyalyi.
\newblock {\em Classical and quantum computation}.
\newblock Number~47. American Mathematical Soc., 2002.

\bibitem{watrous2018theory}
John Watrous.
\newblock {\em The theory of quantum information}.
\newblock Cambridge university press, 2018.

\bibitem{carothers2000real}
Neal~L Carothers.
\newblock {\em Real analysis}.
\newblock Cambridge University Press, 2000.

\bibitem{lin2022lecture}
Lin Lin.
\newblock Lecture notes on quantum algorithms for scientific computation.
\newblock {\em arXiv preprint arXiv:2201.08309}, 2022.

\bibitem{rivlin1990chebyshev}
T.J. Rivlin.
\newblock {\em Chebyshev Polynomials: From Approximation Theory to Algebra and Number Theory}.
\newblock Pure and Applied Mathematics: A Wiley Series of Texts, Monographs and Tracts. Wiley, 1990.

\bibitem{bennett1997strengths}
Charles~H Bennett, Ethan Bernstein, Gilles Brassard, and Umesh Vazirani.
\newblock Strengths and weaknesses of quantum computing.
\newblock {\em SIAM journal on Computing}, 26(5):1510--1523, 1997.

\bibitem{ambainis2000quantum}
Andris Ambainis.
\newblock Quantum lower bounds by quantum arguments.
\newblock In {\em Proceedings of the thirty-second annual ACM symposium on Theory of computing}, pages 636--643, 2000.

\bibitem{martyn2021grand}
John~M Martyn, Zane~M Rossi, Andrew~K Tan, and Isaac~L Chuang.
\newblock Grand unification of quantum algorithms.
\newblock {\em PRX quantum}, 2(4):040203, 2021.

\end{thebibliography}

\section*{End Matter}

\subsection{Proof of Theorem \ref{thm:lower}}
\label{sec:proof_thm3}
\begin{theorem}[Restated]
For algorithms that estimates $\operatorname{Tr}(\rho^k O)$ to additive error $\epsilon$ with constant success probability, there exist quantum state $\rho$ and observable $O$, requires
$$
\Omega\left(\frac{k\|O\|^2}{\epsilon^2}\right)
$$
copies of $\rho$ in the worst case.
\end{theorem}

\begin{proof}
The first step is to reduce the property estimation problem to a state discrimination problem~\cite{helstrom1969quantum}. Consider two states
\begin{equation}
    \rho_0 := |0\rangle\langle0|,\quad
  \rho_1 := (1-\epsilon')|0\rangle\langle0| + \epsilon'|1\rangle\langle1|,
\end{equation}
and the observable $O := |0\rangle\langle0|$. Then
\begin{equation}\label{eq:distinguish}
    \operatorname{Tr}(\rho_0^k O) = 1,\quad
\operatorname{Tr}(\rho_1^k O) = (1-\epsilon')^k,
\end{equation}
where $\epsilon'=\frac{c}{k}$ with $c\in(0,1)$ an absolute constant.
To distinguish these with constant bias, an algorithm must estimate the expectation of the accuracy of the order
$\Delta=1 - (1-\epsilon')^k$.

Suppose we have an algorithm $\mathcal{A}$ that, given $m$ copies of an unknown state $\rho \in\left\{\rho_0, \rho_1\right\}$, outputs an estimate $\hat{\mu}$ such that
$$
\left|\hat{\mu}-\operatorname{Tr}\left(\rho^k O\right)\right| \leq \varepsilon <\Delta/2
$$
with probability at least $2/3$. Then, $\mathcal{A}$ can be served as a discriminator for the two cases:
\begin{itemize}
    \item Run $\mathcal{A}$ on the $m$ copies of the unknown $\rho$ to obtain $\hat{\mu}$.
    \item Output the decision $\hat{b} \in\{0,1\}$ where:
    $\hat{b}=0$ if $\hat{\mu}>\frac{1+\left(1-\epsilon^{\prime}\right)^k}{2};~ \hat{b}=1$ otherwise.
\end{itemize}
Here, we choose the decision threshold to be the midpoint of the two values. The probability of successfully estimating the nonlinear property, then decides the success probability in discriminating whether the state is $\rho_0$ or $\rho_1$.

Now, invoking the results of Ref.~\cite{helstrom1969quantum}, the best success probability for distinguishing $\rho_0$ and $\rho_1$ using $m$ identical copies is given by
\begin{equation}
    \frac{1}{2}+\frac{1}{4}\left\|\rho_0^{\otimes m}-\rho_1^{\otimes m}\right\|_1=\frac{1}{2}+\frac{1}{2}T\left(\rho_0^{\otimes m}, \rho_1^{\otimes m}\right),
\end{equation}
where $T(\cdot, \cdot)$ is the trace distance.
Invoking the Fuchs–van de Graaf inequalities~\cite{fuchs2002cryptographic}: $1-\sqrt{F(\rho, \sigma)} \leq T(\rho, \sigma) \leq \sqrt{1-F(\rho, \sigma)}$, where $F(\rho, \sigma)$ is the fidelity for the input states $\rho$ and $\sigma$, we have the lower bound on the success probability as
\begin{equation}
    1-\frac{1}{2}\sqrt{F\left(\rho_0^{\otimes m}, \rho_1^{\otimes m}\right)}\leq \frac{1}{2}+\frac{1}{2}T\left(\rho_0^{\otimes m}, \rho_1^{\otimes m}\right).
\end{equation}
Direct computation shows that $F\left(\rho_0^{\otimes m}, \rho_1^{\otimes m}\right)=(1-\epsilon')^m$. Therefore, to let $1-\frac{1}{2}\sqrt{F\left(\rho_0^{\otimes m}, \rho_1^{\otimes m}\right)}=1-\frac{1}{2}(1-\epsilon')^{m/2}\geq \frac{2}{3}$, we have $(1-\epsilon')^{m}\leq \frac{4}{9}$. For $k$ to be large, we find $(1-\epsilon')^{m}\approx 1-m\epsilon'=1-\frac{mc}{k}$. This yields $\frac{mc}{k}=\Theta(1)$ only when $m=\Omega(k)$.

We next account for the lower bound due to shot noise (i.e.,~on $\norm{O}$ and $\epsilon$). Our construction is based on distinguishing two probability distributions. Assume the eigendecomposition of $O$ as
\begin{equation}
    O=\|O\| P_{+}-\|O\| P_{-},
\end{equation}
where $P_+$ and $P_-$ are the projectors onto the eigenspace $\norm{O}$ and $-\norm{O}$, respectively. Define two new density operators 
\begin{equation}
\begin{aligned}
    \rho_0:=&\left(\frac{1}{2}+\delta\right) P_{+}+\left(\frac{1}{2}-\delta\right) P_{-},\\ \rho_1:=&\left(\frac{1}{2}-\delta\right) P_{+}+\left(\frac{1}{2}+\delta\right) P_{-}
\end{aligned}
\end{equation}
with $\delta=\frac{\epsilon}{2\|O\|}$. This gives us
\begin{equation}
    \begin{aligned}
    \operatorname{Tr}\left(\rho_0 O\right)&=\|O\|\left[\left(\frac{1}{2}+\delta\right)-\left(\frac{1}{2}-\delta\right)\right]=2 \delta\|O\|=\epsilon,\\
    \operatorname{Tr}\left(\rho_1 O\right)&=\|O\|\left[\left(\frac{1}{2}-\delta\right)-\left(\frac{1}{2}+\delta\right)\right]=-2 \delta\|O\|=-\epsilon.
\end{aligned}
\end{equation}

We now reduce it to a classical two-point problem. We perform the projective measurement in the $\{P_+,P_-\}$ basis on each copy of $\rho$. This converts the quantum problem into the classical task of distinguishing two biased coin distributions on outcomes $\{+1, -1\}$: For $\rho_0$, we obtain `+' with $\frac{1}{2}+\delta$ and `-' with $\frac{1}{2}-\delta$; For $\rho_1$, we obtain `+' with $\frac{1}{2}-\delta$ and `-' with $\frac{1}{2}+\delta$. This results in the difference in the expectation value as
\begin{equation}
    \mathbb{E}_0[X]- \mathbb{E}_1[X]=(2 \delta)-(-2 \delta)=4 \delta .
\end{equation}
Yet, since each coin flip in the quantum case we see is rescaled by the value $\norm{O}$, the difference becomes $4\delta\norm{O}=2\epsilon$, meaning an $\mathcal{O}(\epsilon)$-estimation to the observable. 

The KL divergence between the two coin distributions satisfies
\begin{equation}
    D\left(\frac{1}{2}+\delta \bigg\| \frac{1}{2}-\delta\right)=\mathcal{O}\left(\delta^2\right)=\mathcal{O}\left(\frac{\epsilon^2}{\|O\|^2}\right) .
\end{equation}
Now, we invoke Le Cam’s two-point bound~\cite{lecam1973convergence,tsybakov2009introduction}, which states that any procedure that, with constant success probability, distinguishes $\rho_0^{\otimes m'}$ from $\rho_1^{\otimes m'}$ (equivalently, estimates $\tr{\rho O}$ to error $\epsilon$) must use
\begin{equation}
    m'=\Omega\left(\frac{1}{D}\right)=\Omega\left(\frac{\|O\|^2}{\epsilon^2}\right)
\end{equation}
copies.


\end{proof}

\subsection{BQP-completeness of the nonlinear property estimation problem}

We formulate the decision version of the nonlinear property estimation problem and prove the BQP-completeness of the problem.

\begin{theorem}[BQP-completeness of the decision problem for $\operatorname{Tr}(\rho^k O)$, formal version of Theorem \ref{thm:bqp}]\label{thm:tr-bqp}
For an integer $k\ge 1$, define the promise problem $\mathrm{NonLinear}_k$ as follows. The input is a classical description of:
\begin{enumerate}
\item A polynomial-size (in $n$) quantum circuit $C_{\Gamma}$ that prepares a purification $\ket{\Gamma}_{ASR}$ of an $m$-qubit state $\rho$ via
\[\rho = \operatorname{Tr}_A\big(\ket{\Gamma}\bra{\Gamma}_{ASR}\big),\]
where $A$ is an ancilla register and $SR$ denotes the system on which $\rho$ acts;
\item A succinct classical description of an observable $O$ of operator norm $\norm{O}=\mathcal{O}(1)$; and
\item Thresholds $a,b$ with the promise gap $a-b\ge 1/\operatorname{poly}(n)$, where $n$ denotes the input length.
\end{enumerate}
The promise is that either $\operatorname{Tr}(\rho^k O)\ge a$ (YES) or $\operatorname{Tr}(\rho^k O)\le b$ (NO). Then, $\mathrm{NonLinear}_k$ is BQP-complete.
\end{theorem}

\begin{proof}
The containment in $\mathrm{BQP}$~\cite{kitaev2002classical,watrous2018theory} is seen by the problem is efficiently solved by the algorithm proposed in this work or the generalized swap test by setting the accuracy to be small compared to the promise gap, such as $(a-b)/3$.

We prove the hardness results by a polynomial-time many-one reduction from an arbitrary language $L\in\mathrm{BQP}$. Let $x$ be an instance of $L$ of length $n$. By the definition of BQP there exists a polynomial-size circuit $U_x$ acting on $r(n)$ qubits that prepares a pure state
\begin{equation}
    \ket{\psi_x} = U_x\ket{0^{r(n)}},
\end{equation}
and a single-qubit acceptance projector $\Pi_{\mathrm{acc}}$ (w.l.o.g.~the projector $\ket{1}\bra{1}$ on the first qubit) such that the acceptance probability
\begin{equation}
    p_x := \bra{\psi_x}\Pi_{\mathrm{acc}}\ket{\psi_x}
\end{equation}
satisfies the completeness/soundness gap
\begin{equation}
    p_x \ge 2/3 \quad\text{if } x\in L,\quad p_x \le 1/3 \quad\text{if } x\notin L.
\end{equation}
The strategy is to construct, in polynomial time, an instance of $\mathrm{NonLinear}_k$ whose value of $\tr{\rho^k O}$ encodes $p_x$ that preserves an inverse-polynomial gap.

Specifically, for a a polynomial $q(n)$ (to be specified) in $n$, we set
\begin{equation}\label{eq:lambda-def}
\lambda := 1 - \frac{1}{q(n)k},
\end{equation}
so that $0<\lambda<1$ and $\lambda$ is efficiently computable given $n$ and $k$. Define two orthonormal basis states on a one-qubit register $S$ by $\ket{0}_S,\ket{1}_S$. Let $R$ denote an $r(n)$-qubit register on which $\ket{\psi_x}$ resides, and fix an efficiently preparable state $\ket{\phi}_R$ satisfying $\bra{\phi}\Pi_{\mathrm{acc}}\ket{\phi} = 0$.

Prepare the purification on registers $A\otimes S\otimes R$ by the polynomial-size circuit $C_{\Gamma}$ that outputs
\begin{equation}\label{eq:puri}
\ket{\Gamma_x}_{ASR} = \sqrt{\lambda}\,\ket{0}_A\ket{0}_S\ket{\psi_x}_R + \sqrt{1-\lambda}\,\ket{1}_A\ket{1}_S\ket{\phi}_R.
\end{equation}
Tracing out $A$ yields the mixed state on $SR$:
\begin{equation}\label{eq:rho-decomp}
\rho = \lambda\,\ket{v_1}\bra{v_1} + (1-\lambda)\,\ket{v_2}\bra{v_2},
\end{equation}
where
\[\ket{v_1} := \ket{0}_S\ket{\psi_x}_R,\qquad \ket{v_2} := \ket{1}_S\ket{\phi}_R.
\]
By construction $\bra{v_1}v_2\rangle=0$, we have the spectral decomposition \eqref{eq:rho-decomp} is exact and the eigenvalues of $\rho$ are $\lambda$ and $1-\lambda$.

Define the observable
\begin{equation}\label{eq:O-def}
O := \ket{0}\bra{0}_S \otimes \Pi_{\mathrm{acc}}^{(R)}.
\end{equation}
Then, we obtain
\begin{equation}
    \bra{v_1}O\ket{v_1} = \bra{\psi_x}\Pi_{\mathrm{acc}}\ket{\psi_x} = p_x,\quad \bra{v_2}O\ket{v_2} = 0.
\end{equation}
Since $\rho$ is diagonal in the orthonormal basis $\{\ket{v_1},\ket{v_2}\}$, we have for any integer $k\ge 1$ the result:
\begin{equation}\label{eq:trace-rhok-O}
\operatorname{Tr}(\rho^k O) = \lambda^k\, p_x.
\end{equation}

We now analyze the gap. By construction, if $x\in L$ then $p_x\ge 2/3$; and if $x\notin L$ then $p_x\le 1/3$. Consequently, it yields
\begin{equation}
    \operatorname{Tr}(\rho^k O) \ge \lambda^k\cdot\frac{2}{3} \quad\text{when } x\in L,\quad
\operatorname{Tr}(\rho^k O) \le \lambda^k\cdot\frac{1}{3} \quad\text{when } x\notin L.
\end{equation}
Thus, the separation between YES and NO instances reads
\begin{equation}
    \Delta := \frac{1}{3}\,\lambda^k.    
\end{equation}
It remains only to choose $q(n)$ so that $\Delta\ge 1/\operatorname{poly}(n)$.

Using Bernoulli's inequality~\cite{carothers2000real} $\Big(1 - \frac{1}{t}\Big)^k \ge 1 - \frac{k}{t}$ for $t \ge k \ge 1$ with $k$ integer, we obtain for $t=q(n)k$ that
$$\lambda^k = \Big(1-\frac{1}{q(n)k}\Big)^k \ge 1 - \frac{1}{q(n)}.$$
Therefore, by picking any polynomial $q(n)$ that grows faster than a fixed polynomial (for instance $q(n)=n^{c}$ for sufficiently large constant $c$), we ensure
\begin{equation}
    \Delta = \frac{1}{3}\,\lambda^k \ge \frac{1}{3}\Big(1-\frac{1}{q(n)}\Big) \ge \frac{1}{\operatorname{poly}(n)}.
\end{equation}
As such, it preserves an inverse-polynomial promise gap between the YES and NO cases.

Now, we specify the complexity of the above reduction. The map $x\mapsto (C_{\Gamma},O,k,a,b)$ is computable in polynomial time: i) preparing the circuit $C_{\Gamma}$ requires one single-qubit rotation (to realize the amplitudes $\sqrt{\lambda},\sqrt{1-\lambda}$), ii) one controlled invocation of $U_x$, and iii) fixed state preparation for $\ket{\phi}$. These operations are implementable by polynomial-size circuits. The observable $O$ defined in \eqref{eq:O-def} is also specified succinctly with operator norm $\norm{O}=\mathcal{O}(1)$. The thresholds may be taken as
\begin{equation}
    a := \lambda^k\cdot \frac{2}{3},\qquad b := \lambda^k\cdot \frac{1}{3},
\end{equation}
which satisfy $a-b=\Delta\ge 1/\operatorname{poly}(n)$ as shown above.

Combining the containment in BQP and hardness results, we have the $\mathrm{NonLinear}_k$ problem to be BQP-complete.
\end{proof}

\clearpage

\onecolumngrid
\section*{Supplementary Information}
\appendix

\section{Background}
\label{sec:background}

\subsection{Quantum properties of density matrices and algorithmic tools}

The study of nonlinear properties of quantum states forms a cornerstone of quantum information theory. For a density matrix $\rho$ acting on a Hilbert space $\mathcal{H}$, we often need to evaluate expressions of the form $\tr{\rho^k O}$ where $O$ is an observable and $k$ is a positive integer. These nonlinear functionals encode information that cannot be extracted through simple linear measurements.

To understand why these properties are fundamental, consider the spectral decomposition of a density matrix:
\begin{equation}
\rho = \sum_{i} \lambda_i \ket{\psi_i}\bra{\psi_i},
\end{equation}
where $\lambda_i \geq 0$ are the eigenvalues satisfying $\sum_i \lambda_i = 1$, and $\{\ket{\psi_i}\}$ form an orthonormal basis. When we compute $\rho^k$, we obtain:
\begin{equation}
\rho^k = \sum_{i} \lambda_i^k \ket{\psi_i}\bra{\psi_i}.
\end{equation}

The trace of this quantity, $\tr{\rho^k} = \sum_i \lambda_i^k$, is known as the $k$-th moment of the eigenvalue distribution. For instance, $\tr{\rho^2}$ is the purity, which equals 1 for pure states and is less than 1 for mixed states.

We now introduce some quantum algorithmic tools that will be useful later.
\begin{definition}[Block Encoding]
\label{def:block_encoding}
Let $A \in \mathbb{C}^{2^n \times 2^n}$ be a matrix. A unitary operator $U \in \mathbb{C}^{2^{m+n} \times 2^{m+n}}$ is called an $(\alpha_A, m, \epsilon)$-block encoding of $A$ if
\begin{equation}\label{eq:be}
\left\| A - \alpha_A \left( \bra{0^m} \otimes I_n \right) U_A \left( \ket{0^m} \otimes I_n \right) \right\| \leq \epsilon,
\end{equation}
where $\alpha_A \geq \|A\|$ is the normalization factor, $m$ is the number of ancilla qubits, and $\epsilon \geq 0$ is the encoding error.
\end{definition}
Here, we use the simple definition of the BE that a unitary dilates the target matrix in its upper left corner. The language of BE gives us a simple way to manage (especially non-unitary) matrices or functions of matrices in a unified way.

For our study, two properties of the BE are particularly interesting. The first one is that given the BE of two matrices, we can block encode their product.
\begin{lemma}[Product of multiple block-encoded matrices {\cite[Lemma 53]{gilyen2019quantum}}]\label{lemma:be_prod}
Let $U_A$ and $U_B$ be $(\alpha_A, r, \epsilon)$-block encoding of $n$-qubit operator $A$ and $(\alpha_B, s, \delta)$-block encoding of $n$-qubit operator $B$. Then, $(I_B\otimes U_A)(I_A\otimes U_B)$ is a $(\alpha_A\alpha_B, r+s, \alpha_A\delta+\alpha_B\epsilon)$-block encoding of $AB$.
\end{lemma}

We then introduce the purified quantum query access of a density matrix.
\begin{definition}[Purified quantum query access]\label{def:purified}
Let $\rho \in \mathbb{C}^{N \times N}$ with $N=2^n$ be a density matrix of eigen-decomposition $\rho=\sum_{i=1}^N p_i\ket{\psi_i}\bra{\psi_i}$. It has purified quantum query access if we have access to the unitary oracle $U_\rho$ such that
\begin{equation}\label{eq:purified state}
    U_\rho|0\rangle_E|0\rangle_I=\left|\rho\right\rangle_{E I}=\sum_{i=1}^N \sqrt{p_i}\left|\phi_i\right\rangle_E\left|\psi_i\right\rangle_I, \quad \text { where }\left\langle\phi_i | \phi_j\right\rangle=\left\langle\psi_i | \psi_j\right\rangle=\delta_{i j}.
\end{equation}
The density matrix then acts as the reduced density matrix of the purified state $\operatorname{Tr}_E\left(\left|\rho\right\rangle\left\langle\rho\right|\right)=\rho$.
\end{definition}

The second is that given the purified quantum query access $U_\rho$ of $\rho$, we can construct an error-free BE of $\rho$ with two queries to $U_\rho$ or $U_\rho^\dagger$.
\begin{lemma}[Block encoding of $\rho$ {\cite[Lemma 45]{gilyen2019quantum}}]\label{lemma:be_dm}
Given a $n+a$-qubit unitary $U_\rho$, which is the quantum purified query access of $\rho$. The unitary
$$
\widetilde{U}_\rho=\left(U_\rho^{\dagger} \otimes I_n\right)\left(I_a \otimes \operatorname{SWAP}_n\right)\left(U_\rho \otimes I_n\right)
$$
is an $(1, a+n, 0)$ block encoding of $\rho$.
\end{lemma}

Now, a major tool in our work is the QSVT algorithm. Conceptually speaking, for hermitian matrices, the QSVT can be viewed as performing a function of the matrices. That is, given $\rho$ with its eigen-decomposition defined above, the QSVT algorithm can implement the map
\begin{equation}
    \rho \mapsto p(\rho)=\sum_i p(\rho)=\sum_i p(\lambda_i)\ket{\psi_i}\bra{\psi_i},
\end{equation}
where $p(x)$ is a polynomial function. More formally, provided with BE of $\rho$, one can block encode $p(\rho)$ with the following conditions.
\begin{lemma}[matrix function of real polynomials {\cite[Adapted from Corollary 18]{gilyen2019quantum}}]\label{lemma:qsvt}
Let $U$ be a $(\alpha_A, a, 0)$-block encoding of matrix $A$. Let $p_{\Re} \in \mathbb{R}[x]$ be a degree-$m$ polynomial function such that
\begin{itemize}
    \item $p_{\Re}$ has parity $(m~\text{mod}~2)$ and
    \item $|p_{\Re}|\leq 1,~\forall x\in[-1,1]$.
\end{itemize}
Then, there exists an algorithm construct the $(1,a+1,0)$-block encoding of $p_{\Re}^{\mathrm{(SV)}}(A/\alpha_A)$ with $p_{\Re}^{\mathrm{(SV)}}$ denote the function acting on singular values using $m$ queries to $U$, $U^\dagger$ and $\mathcal{O}((a+1)m)$ other primitive quantum gates.
\end{lemma}

\begin{remark}
For a hermitian matrix $A$, the mapping on the singular value automatically reduces to a mapping on the eigenvalue. As such, the method is sometimes called quantum eigenvalue transformation. This reduction is obvious for positive semi-definite matrices such as the density matrices, as their eigenvalues coincide with the singular values. For general cases, see \cite[Chapter 8]{lin2022lecture} for a detailed discussion.
\end{remark}

\begin{remark}
The result performed here has a parity constraint: the realizable polynomial is either even or odd. While it is possible to lift the constraint by combining even and odd polynomials, we find it unnecessary here as we only need a polynomial function with the same parity as $x^{k-1}$.
\end{remark}

\begin{lemma}[Quantum amplitude estimation, Ref.~\cite{brassard2000quantum}]\label{lemma:ae}
Suppose $U$ is an unitary operation acting on register $a$ and $b$ such that
$$
U|0\rangle_{ab}=\sqrt{p}|0\rangle_a\left|\phi\right\rangle_b+\sqrt{1-p}|1\rangle_a\left|\phi'\right\rangle_b,
$$
where $\left|\phi\right\rangle$ and $\left|\phi'\right\rangle$ are pure quantum states and $p \in[0,1]$. There exists a quantum algorithm that outputs an estimation $\tilde{p} \in[0,1]$ such that
$$
|\tilde{p}-p| \leq \frac{2 \pi \sqrt{p(1-p)}}{K}+\frac{\pi^2}{K^2}
$$
with probability at least $\frac{8}{\pi^2}$, using $O(K)$ queries to $U$ and $U^{\dagger}$.
\end{lemma}

Thus, $\mathcal{O}(\epsilon^{-1})$ queries to $U$ and $U^\dagger$ suffices to estimate $p$ to accuracy $\epsilon$ with constant probability.

\subsection{Chebyshev Polynomials}
We introduce the tools in approximation theory~\cite{sachdeva2014faster,trefethen2019approximation} for approximating a target function using polynomials. The theory aims to provide the near-optimal uniform approximation.
That is, given a function $f \in C[a,b]$ and an approximant $g$ from a prescribed class, one seeks to minimize the supremum norm:
\begin{equation}
\|f-g\|_{\infty} = \sup_{x \in [a,b]} |f(x) - g(x)|.
\end{equation}
The best polynomial approximation is then given by the lowest degree polynomial function with desirable accuracy in the uniform approximation. The constructive aspect of (near) best uniform approximation is obtained by the Chebyshev polynomials of the first kind~\cite{sachdeva2014faster,trefethen2019approximation}, denoted as $T_j(x)$. The function provides (near) optimal performance for functions defined on the canonical interval $[-1,1]$, and can be defined via the trigonometric relation:
$$
T_j(x) = \cos(j\arccos(x)).
$$
To see that it is indeed a polynomial function, one finds that they can also be formulated
through the recursive formulation:
\begin{equation*}
\begin{aligned}
T_0(x) &= 1\\
T_1(x) &= x\\
T_{j+1}(x) &= 2xT_j(x) - T_{j-1}(x), \quad j\geq 1.
\end{aligned}
\end{equation*}
The recurrence relation shows that $T_j(x)$ is indeed a polynomial of degree $j$.

Chebyshev polynomials satisfy the orthogonality relation:
\begin{equation}
\int_{-1}^{1} \frac{T_m(x)T_n(x)}{\sqrt{1-x^2}} dx = 
\begin{cases}
0 & \text{if } m \neq n,\\
\pi & \text{if } m = n = 0,\\
\pi/2 & \text{if } m = n \geq 1.
\end{cases}
\end{equation}

This orthogonality allows us to expand any continuous function $f$ on $[-1,1]$ as:
\begin{equation}
f(x) = \sum_{j=0}^{\infty} c_j T_j(x),
\end{equation}
where the Chebyshev coefficients are given by:
\begin{equation}
c_j = \frac{2-\delta_{j,0}}{\pi} \int_{-1}^{1} \frac{f(x)T_j(x)}{\sqrt{1-x^2}} dx.
\end{equation}


For our application to power functions, we need to understand how Chebyshev expansions behave for functions with singularities. When a function $f$ has an algebraic singularity at the boundary of the interval, the Chebyshev coefficients exhibit specific decay patterns that we exploit for efficient approximation.

\section{Technical Details and Detailed Proofs}
\label{appendix:proofs}

\subsection{Chebyshev expansion of power function}
\label{app:proof_cheby_approx}
In this section, we discuss the approximation of the power function by the Chebyshev polynomial of the first kind. We will establish the existence of a Chebyshev polynomial approximation of $x^k$ with error bounded by $\epsilon$ using a polynomial of degree $m = \Theta(\sqrt{k\log(1/\epsilon)})$.

\subsection{Chebyshev Truncation Analysis}
\label{sec:cheby_trunc}
We determine the appropriate truncation order for our Chebyshev approximation. The Chebyshev expansion of $x^k$ on $[-1,1]$ is:

















\begin{equation}\label{eq:power_cheby_expansion}
x^k = 2^{1-k}\sum_{j=0}^{\lfloor k/2 \rfloor} \alpha_j \binom{k}{j} T_{k-2j}(x),
\end{equation}
where $\alpha_j = \frac{1}{2}$ if $k$ is even and $j = k/2$, and $\alpha_j = 1$ otherwise.

To approximate this function, we truncate the expansion and keep only those Chebyshev modes with index $|n| \leq m$, i.e.,
\begin{equation}
P_m(x) = \sum_{j: |k-2j| \leq m} 2^{1-k}\alpha_j\binom{k}{j} T_{k-2j}(x).
\end{equation}
The error is then given by
\begin{equation}
R_m(x) = x^k - P_m(x) = \sum_{j: |k-2j| > m} 2^{1-k}\alpha_j\binom{k}{j} T_{k-2j}(x).
\end{equation}
Since $|T_n(x)| \leq 1$ for $x \in [-1,1]$, we can bound this error as
\begin{equation}
\|R_m\|_\infty \leq \sum_{j: |k-2j| > m} 2^{1-k}\alpha_j\binom{k}{j}.
\end{equation}

For the probability interpretation, note that when $k$ is odd, all $\alpha_j = 1$ and we have exactly $\text{Pr}[\left|\text{Bin}(k, \frac{1}{2}) - \frac{k}{2}\right| > \frac{m}{2}]$. When $k$ is even, the adjustment factor $\alpha_{k/2} = \frac{1}{2}$ affects only the boundary case, which has a negligible impact on the tail bound for large $k$.

Therefore, we can still use the approximation:
\begin{equation}
\|R_m\|_\infty \lesssim \text{Pr}\left[\left|\text{Bin}(k, \frac{1}{2}) - \frac{k}{2}\right| > \frac{m}{2}\right].
\end{equation}
Using a two-sided Chernoff bound, we obtain
\begin{equation}
\text{Pr}\left[\left|\text{Bin}\left(k, \frac{1}{2}\right) - \frac{k}{2}\right| > d\right] \leq 2 \exp\left(-\frac{2d^2}{k}\right).
\end{equation}
Subsequently setting $d = \frac{m}{2}$ and requiring this probability to be at most $\epsilon$, we get:
\begin{equation}\label{eq:cheby_trunc}
2 \exp\left(-\frac{m^2}{2k}\right) \leq \epsilon.
\end{equation}
By taking logarithms, we have $\ln(2) - \frac{m^2}{2k} \leq \ln(\epsilon)$, which implies
\begin{equation}
m^2 \geq 2k(\ln(2) - \ln(\epsilon)) = 2k\ln\left(\frac{2}{\epsilon}\right).
\end{equation}
Therefore, we finally have
\begin{equation}
m \geq \sqrt{2k\ln\left(\frac{2}{\epsilon}\right)} = \mathcal{O}\left(\sqrt{k\ln(1/\epsilon)}\right).
\end{equation}

We find the approximation polynomial function has parity $k~\mathrm{mod}~2$ due to Eq.~\eqref{eq:power_cheby_expansion}. We summarize the results as follows.
\begin{lemma}\label{lemma:cheby_approx}
Given power function $f(x)=x^k$ for $x\in[-1,1]$ and $\epsilon\in(0,1)$, there exists an degree-$m$ polynomial function $p_m(x)$ such that 
\begin{equation}
    \sup_{x\in[-1,1]}|f(x)-p_m(x)|\leq \epsilon,
\end{equation}
such that
\begin{equation}
    m=\sqrt{2k\ln(2/\epsilon)}.    
\end{equation}
Besides, $p_m(x)$ has parity $k~\mathrm{mod}~2$.
\end{lemma}

\subsection{Proof of asymptotic optimality and origin of quadratic speedup}
\label{sec:proof_optimality}
In this section, we prove that the approximation degree provided in the last section is optimal. This is not surprising, as the Chebyshev polynomial typically provides (near) best performance.
We prove the optimality by establishing the lower bound on the polynomial degree required for approximating $x^k$. 

Besides, we claim that the quadratic reduction in polynomial degree for approximating $x^k$ can be understood through the perspective of Bernstein's inequality. This classical result in approximation theory reveals a fundamental difference between polynomial behavior at the interior versus the boundary of an interval, which explains what enables our speedup.

\subsubsection{Bernstein's Inequality: Statement and Intuition}
We begin with the precise statement of Bernstein's inequality, which constrains the rate of change of polynomials.

\begin{lemma}[Bernstein's Inequality]
\label{thm:bernstein}
Let $P(x)$ be a polynomial of degree $\ell$ satisfying $|P(x)| \leq 1$ for all $x \in [-1,1]$. Then:
\begin{enumerate}
\item For any $x \in (-1,1)$: $|P'(x)| \leq \frac{\ell}{\sqrt{1-x^2}}$,
\item For the endpoints: $|P'(\pm 1)| \leq \ell^2$.
\end{enumerate}
Moreover, the bound at the endpoints is achieved by the Chebyshev polynomial $T_\ell(x)$.
\end{lemma}

To understand why this inequality is remarkable, consider what it tells us about polynomial derivatives. At first glance, one might expect that if a polynomial of degree $\ell$ is bounded by 1, its derivative should be bounded by something like $\ell$ (since differentiating reduces the degree by 1). However, Bernstein's inequality shows that the derivative can grow as $\ell^2$ at the endpoints.

This quadratic factor arises from a fundamental property of polynomials: they can oscillate more rapidly near the boundaries of an interval while remaining bounded. The Chebyshev polynomials exemplify this behavior—they achieve their maximum slope precisely at $x = \pm 1$~\cite{trefethen2019approximation,rivlin1990chebyshev}.
To understand where the $\ell^2$ factor comes from, let us examine the proof for the endpoint case. Consider a polynomial $P(x)$ of degree $\ell$ with $|P(x)| \leq 1$ on $[-1,1]$.

The Chebyshev polynomial $T_\ell(x) = \cos(\ell \arccos(x))$ satisfies $|T_\ell(x)| \leq 1$ on $[-1,1]$. Its derivative at $x = 1$ can be computed using the chain rule:
\begin{equation}
\begin{aligned}
T'_\ell(x) &= \ell \sin(\ell \arccos(x)) \cdot \frac{1}{\sqrt{1-x^2}},\\
T'_\ell(1) &= \lim_{x \to 1^-} \ell \sin(\ell \arccos(x)) \cdot \frac{1}{\sqrt{1-x^2}}.
\end{aligned}
\end{equation}
By direct computation for Chebyshev polynomials, we find:
\begin{equation}
T'_\ell(1) = \ell^2.
\end{equation}

A deep result in approximation theory states that among all monic polynomials of degree $\ell$, the polynomial $2^{1-\ell}T_\ell(x)$ has the smallest maximum absolute value on $[-1,1]$. This extremal property extends to derivatives: no polynomial of degree $\ell$ bounded by 1 on $[-1,1]$ can have a derivative at $x = 1$ larger than $\ell^2$.

The quadratic scaling can be understood through the lens of polynomial interpolation. A polynomial of degree $\ell$ is determined by $\ell + 1$ values. When we constrain $|P(x)| \leq 1$ on $[-1,1]$, the polynomial can exploit all $\ell + 1$ degrees of freedom to create rapid oscillations that accumulate near the boundary, leading to a derivative that scales as $\ell^2$ rather than $\ell$.

\subsection{Lower bound on the polynomial degree}

It is well known that the Chebyshev polynomial gives the best uniform approximation of the power function $f(x)=x^k$ for the truncation degree $d=k-1$~\cite{rivlin1990chebyshev}, i.e., achieving the smallest possible approximation error for the prescribed degree among all polynomial functions. For $d<k-1$, it only provides near-best results. In this section, we show that while the Chebyshev expansion may not be the best approximation, its truncation degree matches the asymptotic lower bound and could only differ by a subleading factor from the best uniform approximation.

To this end, we resort to the classical results that study the relation between the Chebyshev approximation error $S_d(f)$ and the best uniform approximation error $E_d(f)$ of both degrees $d$:
\begin{equation}\label{eq:approx_error}
    S_d(f):=\norm{f-p_d(f)}_\infty,\quad E_d(f):=\norm{f-q_d(f)}_\infty,
\end{equation}
where $p_d$ and $q_d$ are degree-$d$ Chebyshev and best uniform approximation, respectively.

The relation to the two errors is given by
\begin{lemma}[{\cite[Theorem 3.3]{rivlin1990chebyshev}}]\label{lemma:cheby_best}
The Chebyshev and best uniform approximation errors as given by Eq.~\eqref{eq:approx_error} satisfy
\begin{equation}
    E_d(f) \leqslant S_d(f)<\left(4+\frac{4}{\pi^2} \log(d)\right) E_d(f).
\end{equation}
\end{lemma}

Using this result, we obtain the following degree lower bound in approximating the power function.
\begin{theorem}[Lower Bound for Polynomial Approximation of $x^k$]
For any positive integer $k$ and any $\epsilon \in (0,1]$, if $p(x)$ is a polynomial of degree at most $d$ with
\begin{equation}
    \max_{x \in [-1,1]} |x^k - p(x)| \leq \epsilon,
\end{equation}
then
\begin{equation}
    d=\Omega(\sqrt{k\log(1/\epsilon)}).
\end{equation}
\end{theorem}
\begin{proof}
From Eq.~\eqref{eq:cheby_trunc}, we know that the Chebyshev approximation error satisfies
\begin{equation}
    S_d(f)=2 \exp\left(-\frac{d^2}{2k}\right).
\end{equation}
Besides, from Lemma \ref{lemma:cheby_best}, we have
\begin{equation}
    \frac{\pi^2}{4\pi^2+4\log(d)}S_d(f)<E_d(f),
\end{equation}
which indicates a lower bound given by the Chebyshev approximation error. Then, let the left-hand side of the equation be no more than $\epsilon$, we have
\begin{equation}
    \frac{2\pi^2}{4\pi^2+4\log(d)}\exp\left(-\frac{d^2}{2k}\right)\leq \epsilon.
\end{equation}
The solution when taking the equality is given by
\begin{equation}\label{eq:exact_d}
    d^2=2 k \ln \left(\frac{\pi^2 / 2}{\epsilon\left(\pi^2+\ln d\right)}\right)=2 k\left[\ln \frac{\pi^2}{2 \epsilon}-\ln \left(\pi^2+\ln d\right)\right],
\end{equation}
which is a transcendental equation, meaning no closed form solution for $d$. Now, when $\ln d$ is comparable to $\pi^2$, we take $\pi^2+\ln d=c\pi^2$ for some absolute constant $c$. Then, we obtain
\begin{equation}
    d\geq \sqrt{2k\ln(1/(2c\epsilon))}=\Omega\left(\sqrt{k\log(1/\epsilon)}\right).
\end{equation}

When the assumption on $\ln d$ is violated, implying $\epsilon$ could be very small, we first observe that the leading order solution when ignoring the slowly varying $-\ln \left(\pi^2+\ln d\right)$ term is given by
\begin{equation}
    d_0=\sqrt{2k\ln(\pi^2/2\epsilon)}.
\end{equation}
Then, we consider the first-order correction given by the Taylor expansion of the logarithm:
\begin{equation}
    \ln(\pi^2+\ln d_0)=\ln(\pi^2+\mathcal{O}(\ln \ln(1/\epsilon)))=\mathcal{O}(\ln\ln(1/\epsilon)).
\end{equation}
Thus, it results in
\begin{equation}
    d\gtrsim\sqrt{2 k \ln(\pi^2/2\epsilon)}\left(1-\mathcal{O}\left(\frac{\ln \ln (1 / \epsilon)}{\ln (1 / \epsilon)}\right)\right) =\Omega\left(\sqrt{k\log(1/\epsilon)}\right).
\end{equation}

In both scenarios, we obtain the claimed result, which concludes the proof.
\end{proof}

We summarize the optimal degree of approximating the power function as follows.
\begin{corollary}\label{cor:optimal_power}
For $f(x)=x^k:[-1,1]\mapsto[-1,1],k\in\mathbb{Z}^+$ and $\epsilon\in(0,1]$, the optimal asymptotic degree $d$ of polynomial function approximate $f(x)$ to error $\epsilon$ in uniform approximation is given by 
\begin{equation}
    d=\Theta\left(\sqrt{k\log(1/\epsilon)}\right).
\end{equation}
\end{corollary}

\subsection{Query complexity lower bound}

We now connect the lower bound on the degree of a scalar function to that of the query complexity of a matrix function. 
Consider the block encoding $U_A$ of matrix $A$ in the form 
\begin{equation}\label{eq:a_be}
    U_A=\left(\begin{array}{cc}
        A &  *\\
        * & *
    \end{array}\right),
\end{equation}
we want to construct the block encoding of $f(A)$ with error $\epsilon$ with $f(x)$ some function. The goal is to determine the necessary query times $T$ of applying $U_A$ and $U_A^\dagger$.

We first note that the query complexity lower bound is discussed in \cite[Theorem 73]{gilyen2019quantum}. That is for $f: I\subseteq[-1,1]\mapsto\mathbb{R}$ and $x \neq y \in I \cap\left[-\frac{1}{2}, \frac{1}{2}\right]$, we have 
\begin{equation}\label{eq:lower_query1}
    T=\Omega\left(\frac{|f(x)-f(y)|-2 \epsilon}{|x-y|}\right).
\end{equation}
More precisely, for $x\neq y\in I$, the non-asymptotic lower bound is given by
\begin{equation}
    T \geq \frac{\max \left[f(x)-f(y)-2 \epsilon, \sqrt{1-(f(y)-\epsilon)^2}-\sqrt{1-(f(x)+\epsilon)^2}\right]}{\sqrt{2} \sqrt{1-x y-\sqrt{\left(1-x^2\right)\left(1-y^2\right)}}}.
\end{equation}
The lower bound on the query complexity relates to the maximal derivative of the function.

We now give a more sophisticated result that relates the lower bound to the degree of best approximation based on recent breakthroughs. To this end, we introduce the following lemma.
\begin{lemma}[Optimal query complexity of matrix function {\cite[Proposition 1.12]{montanaro2024quantum}}]\label{lemma:optimal_qeury}
Let $f(x):[-1,1] \mapsto[-1,1]$ be a continuous function, and $\epsilon \in(0,1]$ a constant. Let $U_A$ be a unitary in the form Eq.~\eqref{eq:a_be} that dilates $A$, where $A$ is a $\mathcal{O}(1)$-sparse Hermitian matrix with $\|A\| \leq 1-\delta$ for some constant $\delta$. Then, there exists a unitary that dilates an $\epsilon$-approximation to $f(A)$ using 
\begin{equation}
    \Theta\left(\widetilde{\operatorname{deg}}_{\epsilon}(f)\right),\quad \widetilde{\operatorname{deg}}_{\epsilon}(f):=
\min \left\{d: \inf _{p_d} \sup _{x \in[-1,1]}\left|f(x)-p_d(x)\right| \leq \epsilon\right\}
\end{equation}
queries to $U_A$ and $U_A^\dagger$ with $p_d\in\mathbb{R}[x]$ a degree-$d$ polynomial function.
\end{lemma}
We note that the requirement on the operator norm such that $\|A\| \leq 1-\delta$ is a condition for efficiently block encoding $A$ exploiting the sparse structure. For convenience, we simply let $\delta=0$ in the following and obtain the following result on query complexity for approximating the power function. 
\begin{corollary}[Optimal query complexity of matrix power function]\label{cor:optimal_qeury_power}
Let $f(x):=x^k, k\in\mathbb{Z}^+$, and $\epsilon \in(0,1]$ a constant. Let $U_A$ be a unitary that dilates $A$, where $A$ is a $\mathcal{O}(1)$-sparse Hermitian matrix with $\|A\| \leq 1$. Then, there exists a unitary that dilates an $\epsilon$-approximation to $f(A)$ using 
\begin{equation}
    \Theta\left(\sqrt{k\log(1/\epsilon)}\right)
\end{equation}
queries to $U_A$ and $U_A^\dagger$ with $p_d\in\mathbb{R}[x]$ a degree-$d$ polynomial function.
\end{corollary}
\begin{proof}
From Lemma \ref{lemma:optimal_qeury}, we know that the optimal query complexity is given by the optimal degree in approximating the targeted power function, which is given by Corollary \ref{cor:optimal_power}. This completes the proof.
\end{proof}

\begin{remark}
We note that the bound on the query complexity shown in Corollary \ref{cor:optimal_qeury_power} holds for general density matrices as long as one can efficiently block encode them, such as through the purified quantum query access. This is because for the upper bound side, using the Chebyshev expansion and QSVT algorithm given by Lemma \ref{lemma:qsvt}, we can reach the same asymptotic scaling. For the lower bound side, we note that the $\mathcal{O}(1)$ sparsity condition does not always hold for density matrices. Yet, it suffices to prove that there exist cases in which the condition holds. As such, we note that the $\mathcal{O}(1)$-sparse density matrices are strictly contained within the set of all density matrices; the lower bound immediately carries over to the full class of density operators. Hence, we obtain the tight bound
\begin{equation}
    \Theta\left(\sqrt{k\log(1/\epsilon)}\right)
\end{equation}
on the query complexity for density matrices admitting purified quantum query access.
\end{remark}

We now provide a result on the query complexity for estimating the 
\begin{lemma}\label{lemma:property_lower_bound}
Let $\ket{\psi}$ be a $n$-qubit input state and $\epsilon\in(0,1]$ a constant. Let $O$ be an $n$-qubit Hermitian observable with operator norm $\norm{O}>1$. Let $U_A$ be a $(1,a,0)$-block encoding of $A$. Then, there exist cases that at least
\begin{equation}
    \Omega\left(\frac{\norm{O}}{\epsilon}\right)
\end{equation}
queries to $U_A$ are necessary to output an $\epsilon$-estimation to $\tr{\ket{\psi}\bra{\psi}AO}$ with a constant success probability.
\end{lemma}
\begin{proof}
Our proof strategy is similar to that of the query complexity lower bound of the unstructured search using a hybrid argument~\cite{bennett1997strengths,ambainis2000quantum}. That is, the general circuit structure for such estimation is given by an interchangeably application of $U_A$ and some other unitary $V$; and then some final measurement related to the observable $O$. Our proof then relies on constructing two different block-encoded matrices $A_0$ and $A_1$ such that distinguishing between the two cases remains hard unless a certain query time is met.

Denote $\ket{\nu}$ as the state corresponding to the largest singular value of $O$, such that $|\bra{\nu}O\ket{\nu}|=\norm{O}$. We prescribe the input state to be $\ket{\psi}\equiv\ket{\nu}$. The two cases of the block-encoded matrices are $A_0=0$ and $A_1=\delta\ket{\nu}\bra{\nu}$ with $\delta=\frac{2\epsilon}{\norm{O}}$. Without loss of generality, we follow the reflection convention of the block encoding unitaries, which give
\begin{equation}
    U_0=\left(\begin{array}{cc}
        0 & I \\
        I & 0
    \end{array}\right),\quad
    U_1=\left(\begin{array}{cc}
        \delta\ket{\nu}\bra{\nu} & \sqrt{1-\delta^2\ket{\nu}\bra{\nu}} \\
        \sqrt{1-\delta^2\ket{\nu}\bra{\nu}} & -\delta\ket{\nu}\bra{\nu}
    \end{array}\right).
\end{equation}
For different conventions of the block encoding form, see Ref.~\cite{martyn2021grand}.
The expectation values of the two cases are given by 
\begin{equation}
    E_0:=\tr{\ket{\nu}\bra{\nu}A_0O}=0,\quad E_1:=\tr{\ket{\nu}\bra{\nu}A_1O}=2\epsilon.
\end{equation}
Any algorithm that estimates $E_i,~i\in\{0,1\}$ to error $\epsilon$ given access to $U_i$ distinguishes between these two cases with constant success probability.

The two sets of states after $t$ queries of the BE unitaries are
\begin{equation}
\begin{aligned}
    \ket{\psi_0(t)}=\prod_{j=1}^t V_jU_0\ket{\nu}\ket{0^{a+b}}\\
    \ket{\psi_1(t)}=\prod_{j=1}^t V_jU_1\ket{\nu}\ket{0^{a+b}}
\end{aligned},
\end{equation}
where $V_j$ is some unitary acting on $(n+a+b)$ qubits. 

Next, we claim that the trace distance between the two sets of states after $T$ iterations satisfies
\begin{equation}
    \norm{\ket{\psi_0(T)}-\ket{\psi_1(T)}}_1\leq T\norm{U_0-U_1}.
\end{equation}
This is shown by induction. The base case is readily verified 
\begin{equation}
    \| V_0 U_0\ket{\nu}\ket{0^{a+b}}-V_0 U_1\ket{\nu}\ket{0^{a+b}}\|_1=\| U_0\ket{\nu}\ket{0^{a+b}}-U_1\ket{\nu}\ket{0^{a+b}}\|_1\leq\| U_0-U_1 \|,
\end{equation}
where we have applied H\"{o}lder's inequality.
Now, assume for the $(T-1)$ step, the result holds, and we deduce for the next step:
\begin{equation}
\begin{aligned}
    \| V_T U_0\ket{\psi_0(T-1)}-V_T U_1\ket{\psi_1(T-1)}\|_1&=\|  U_0\ket{\psi_0(T-1)}-U_1\ket{\psi_0(T-1)}+U_1\ket{\psi_0(T-1)}- U_1\ket{\psi_1(T-1)}\|_1\\
    &\leq \|  U_0\ket{\psi_0(T-1)}-U_1\ket{\psi_0(T-1)}\|_1 + \|U_1\ket{\psi_0(T-1)}- U_1\ket{\psi_1(T-1)}\|_1\\
    &=\norm{U_0-U_1} + \norm{\ket{\psi_0(T-1)}- \ket{\psi_1(T-1)}}_1=T\norm{U_0-U_1}.
\end{aligned}
\end{equation}
Next, the $\norm{U_0-U_1}$ is computed by solving the eigenvalues of $U_0-U_1$. Through basic computation, we have 
\begin{equation}
    \left\|U_B-U_A\right\|=\sqrt{\delta^2+\left(1-\sqrt{1-\delta^2}\right)^2} =\sqrt{\delta^2+\mathcal{O}\left(\delta^4\right)}=\delta \sqrt{1+\mathcal{O}\left(\delta^2\right)}=\delta+\mathcal{O}\left(\delta^3\right),
\end{equation}
where we have used $1-\sqrt{1-\delta^2}=1-\left(1-\frac{\delta^2}{2}-\frac{\delta^4}{8}+\mathcal{O}\left(\delta^6\right)\right)=\frac{\delta^2}{2}+\frac{\delta^4}{8}+\mathcal{O}\left(\delta^6\right) .$
This suggests that 
\begin{equation}
    \norm{\ket{\psi_0(T)}-\ket{\psi_1(T)}}_1\lesssim T\delta = \frac{2T\epsilon}{\norm{O}}.
\end{equation}
This means that to achieve a constant success probability over random guessing, the query times are lower bounded by
\begin{equation}
    T=\Omega\left(\frac{\norm{O}}{\epsilon}\right).
\end{equation}

Finally, as the observable expectation value estimation problem with the prescribed accuracy solves the distinguish problem, which implies the same lower bound follows.
\end{proof}

\begin{remark}
The above results are considered error-free BE. When given a noisy BE of $U_i,~i\in\{0,1\}$, the error in BE is amplified by a factor of $\norm{O}$ in $E_i$. This means that when the error in the BE is above $\frac{\epsilon}{\norm{O}}$, we cannot distinguish between the two cases.
\end{remark}

We close this section by providing results on the lower bound of the query complexity of the nonlinear property estimation problem.
\begin{theorem}[Theorem 2 restated]
Let $\ket{\rho}_{EI}=U_\rho\ket{0}_E\ket{0}_I$ be a purified (input) state of $\rho$. Let $\epsilon\in(0,1]$ be a constant. Then, 
\begin{equation}
    \Omega\left(\frac{\sqrt{k\log(\norm{O}/\epsilon)}\norm{O}}{\epsilon}\right)
\end{equation}
queries to $U_\rho$ and $U_\rho^\dagger$ is needed to output and $\epsilon$-approximation to $\tr{\rho^k O}$.
\end{theorem}
\begin{proof}
The first step is to construct a BE unitary $U_{g(\rho)}$ of $g(\rho)$, which is a $\frac{\epsilon}{2\norm{O}}$-approximation to the operator $\rho^{k-1}$. This entails an $\Omega(\sqrt{k\log(\norm{O}/\epsilon)})$ queries to $U_\rho$ and $U_\rho^\dagger$ as dictated by Corollary \ref{cor:optimal_qeury_power}. Then, with input state $\ket{\rho}$ and BE of $g(\rho)$, Lemma \ref{lemma:property_lower_bound} demands an $\Omega\left(\frac{\norm{O}}{\epsilon}\right)$ queries to $U_{g(\rho)}$. Combining the two results, we obtain the claim lower bound.
\end{proof}

\section{Algorithm Formulation}
We use quantum amplitude estimation to estimate the nonlinear properties of the density matrix. The first step is to construct a block encoding of a nonlinear function of the density matrix that approximates $\rho^{k-1} O$. We denote this BE unitary as $U_p$ such that $
(\langle 0|_b \otimes I) U_{p} (|0\rangle_b \otimes I) = p(\rho)O=:p'(\rho)$, where $\|p(\rho)-\rho^k\|\leq \epsilon'$ for some $\epsilon'$ to be determined. The circuit construction is similar to the Hadamard test circuit, involving a control qubit $c$, an ancilla register $b$, and registers $E$ and $I$. 
We initialize the state as $|0\rangle_c |0\rangle_b |\rho\rangle_{E I}$ and proceed as the following:
\begin{equation}\label{eq:circ_hadamard_test}
\begin{alignedat}{2}
\ket{0}_c\ket{0}_b\ket{\rho}_{EI} &\xrightarrow{H_c} &\quad& \frac{|0\rangle_c + |1\rangle_c}{\sqrt{2}} \otimes |0\rangle_b |\rho\rangle_{EI} \\
&\xrightarrow{\rm{c}-U_{p}} &\quad& \frac{1}{\sqrt{2}} \left[ |0\rangle_c |0\rangle_b |\rho\rangle_{E I} + |1\rangle_c \left( |0\rangle_b p'(\rho) |\rho\rangle_{E I} + |\Phi^\perp\rangle_{b E I} \right) \right] \\
&\xrightarrow{H_c} &\quad& \quad~|0\rangle_c \otimes \frac{1}{2} \left[ |0\rangle_b |\rho\rangle_{E I} + |0\rangle_b p'(\rho) |\rho\rangle_{E I} + |\Phi^\perp\rangle_{b E I} \right] \\
&&\quad& + |1\rangle_c \otimes \frac{1}{2} \left[ |0\rangle_b |\rho\rangle_{E I} - |0\rangle_b p'(\rho) |\rho\rangle_{E I} - |\Phi^\perp\rangle_{b E I} \right],
\end{alignedat}
\end{equation}
where $p(\rho)O$ acts on register $I$, the $H_c$ is a Hadamard gate acting on the register $c$, $p'(\rho)$ acts on register $I$, $|\Phi^\perp\rangle_{b E I}$ is the state that orthogonal to $\ket{0}_b$ and $\mathrm{c}-U_{p}$ is the controlled unitary $U(p)$ controlled by register $c$.

Subsequently, the probability of measuring $\ket{0}_c$ is
\begin{equation}
\begin{aligned}
P(|0\rangle_c) &= \frac{1}{4} \left(\left\| |\rho\rangle_{E I} + p'(\rho) |\rho\rangle_{E I} \right\|_2^2 +  \langle \Phi^\perp | \Phi^\perp \rangle\right)\\
&=\frac{1}{4} \left(\left(1+ \langle \rho | (p'(\rho))^\dagger p'(\rho) |\rho \rangle + \langle \rho | p'(\rho) |\rho \rangle + \langle \rho | (p'(\rho))^\dagger |\rho \rangle \right) + \left(1- \langle \rho | (p'(\rho))^\dagger p'(\rho) |\rho \rangle\right) \right)\\
&=\frac{1}{2}+ \frac{\mathrm{Re}(\bra{\rho} p(\rho)O\ket{\rho})}{2}=\frac{1}{2}+ \frac{\mathrm{Re}(\tr{\rho p(\rho)O})}{2} \approx \frac{1}{2}+ \frac{\mathrm{Re}(\tr{\rho^k O})}{2},
\end{aligned}
\end{equation}
where $\|\cdot\|_2$ is the $\rm{L}_2$ norm. In the second line, we have used that $\left\||0\rangle_b p'(\rho) |\rho\rangle_{E I} + |\Phi^\perp\rangle_{b E I}\right\|_2=1$ due to the normality of the state. This yields the estimator of the real part of the nonlinear property as
\begin{equation}\label{eq:ae_real}
    \text{Re} \left( \text{Tr}\left( \rho p(\rho)O \right) \right) = 2 P(|0\rangle_c) - 1.
\end{equation}
The imaginary part follows the same logic as the Hadamard test construction.
We implement a phase gate $S^\dagger$ before the second Hadamard gate. Denoting the probability of measuring $\ket{0}_c$ as $P'(\ket{0}_c)$. The estimator is the same as the real case:
\begin{equation}\label{eq:ae_img}
    \text{Im} \left( \text{Tr}\left( \rho p(\rho)O \right) \right) = 2 P'(|0\rangle_c) - 1.
\end{equation}

    

Finally, we present the result on (approximately) BE of the operator $\rho^{k-1}$ and $\rho^{k-1}O$.
\begin{lemma}\label{lemma:be_power}
Let $\widetilde{U}_\rho$ be $(1, a, 0)$-block encoding unitary of $\rho$. Then, there exists quantum circuit constructs an $(1, a+1, \epsilon)$-block encoding of $p_m(\rho)$ with $p_m(\cdot)$ a degree $m$ polynomial function such that 
\begin{equation}
    \|p_m(\rho)-\rho^{k-1}\|\leq \epsilon,
\end{equation}
using $m=\sqrt{2(k-1)\ln(2/\epsilon)}$ queries to $\widetilde{U}_\rho$, $\widetilde{U}_\rho^\dagger$, and $\mathcal{O}(m(a+1))$ other primitive quantum gates.
\end{lemma}
\begin{proof}
Our first step is to construct the $(1, a+n+1, \epsilon)$-block encoding of $p_m(\rho)$. To this end, Lemma \ref{lemma:cheby_approx} gives the polynomial function $p_m(x)$ such that $\sup_{x\in[-1,1]}|p_m(x)-x^{k-1}|\leq \epsilon$ with the claimed degree $m$. Also, $p_m(x)$ has parity $(k-1)~\mathrm{mod}~2$. As such, invoking Lemma \ref{lemma:qsvt}, we obtain a $(1,a+1,\epsilon)$-BE of $p_m(\rho)$ by noting that $\|p_m(\rho)-\rho^{k-1}\|=|\sup_{x\in[-1,1]}|p_m(x)-x^{k-1}|\leq \epsilon$ with $m$ queries to $\widetilde{U}_\rho$ and $\widetilde{U}_\rho^\dagger$.
\end{proof}

\begin{corollary}\label{cor:be_power}
Let $\widetilde{U}_\rho$ be $(1, a, 0)$-block encoding unitary of $\rho$. Let ${U}_O$ be $(\alpha_O, b, 0)$-block encoding unitary of $O$ such that $\alpha_O\geq\norm{O}$. Then, there exists quantum circuit constructs an $(\alpha_O, a+b+1, \alpha_O\epsilon)$-block encoding of $p_m(\rho)O$ with $p_m(\cdot)$ a degree $m$ polynomial function such that 
\begin{equation}
    \|p_m(\rho)-\rho^{k-1}\|\leq \epsilon,
\end{equation}
using $m=\sqrt{2(k-1)\ln(2/\epsilon)}$ queries to $\widetilde{U}_\rho$, $\widetilde{U}_\rho^\dagger$, one query to $U_O$ and $\mathcal{O}(m(a+b+1))$ other primitive quantum gates.
\end{corollary}
\begin{proof}
By Lemma \ref{lemma:be_power}, we have the BE of the (approximate) power function of $\rho$. We get the BE of $p_m(\rho)O$ by taking the product of $p_m(\rho)$ and $O$ using Lemma \ref{lemma:be_prod}, yielding the pronounced results.
\end{proof}

\end{document}